\documentclass[%
reprint,
superscriptaddress,
nofootinbib, notitlepage,
amsmath,amssymb,
aps, pra, onecolumn,
longbibliography
]{revtex4-1}
\usepackage[colorlinks=true,urlcolor=blue,citecolor=blue,linkcolor=blue]{hyperref}
\usepackage{amsmath, amssymb, amsthm, amsfonts, latexsym, graphicx}
\usepackage{float}
\usepackage{verbatim}
\usepackage{thm-restate}
\usepackage{url}
\usepackage{complexity}
\usepackage{hypcap}
\usepackage{algorithm}
\usepackage{algorithmic}
\usepackage{tikz}

\def\ket#1{| #1 \rangle}
\def\bra#1{\langle #1 |}

\def\kb#1#2{| #1 \rangle\!\langle #2 |}

\def\cG{\mathcal{G}}
\def\cH{\mathcal{H}}

\def\cL{\mathcal{L}}

\newtheorem{theorem}{Theorem}
\newtheorem{corollary}{Corollary}
\newtheorem{lemma}{Lemma}
\newtheorem{definition}{Definition}
\def\eq#1{Eq.~\eqref{eq:#1}}

\def\eq#1{Eq.~\eqref{eq:#1}}

\def\ket#1{| #1 \rangle}
\def\bra#1{\langle #1 |}
\def\ketbra#1#2{|#1\rangle\!\langle#2|}
\def\bracket#1#2{\langle #1 | #2 \rangle}
\def\kb#1#2{| #1 \rangle\!\langle #2 |}

\def\eq#1{Eq.~\eqref{eq:#1}}

 \def\comment#1{}		
  \def\check#1{#1}		
\usepackage{qtree} 

\begin{document}
\title{Quantum Language Processing}
\author{Nathan Wiebe}
\affiliation{Microsoft Research, One Microsoft Way, Redmond WA 98052}
\author{Alex Bocharov}
\affiliation{Microsoft Research, One Microsoft Way, Redmond WA 98052}
\author{Paul Smolensky}
\affiliation{Dept of Cognitive Science Johns Hopkins University, Baltimore MD 21218}
\affiliation{Microsoft Research, One Microsoft Way, Redmond WA 98052}
\author{Matthias Troyer}
\affiliation{Microsoft Research, One Microsoft Way, Redmond WA 98052}
\author{Krysta M. Svore}
\affiliation{Microsoft Research, One Microsoft Way, Redmond WA 98052}
\begin{abstract}
We present a representation for linguistic structure that we call a Fock-space representation, which allows us to embed problems in language processing into small quantum devices. We further develop a formalism for understanding both classical as well as quantum linguistic problems and phrase them both as a Harmony optimization problem that can be solved on a quantum computer which we show is related to classifying vectors using quantum Boltzmann machines.  We further provide a new training method for learning quantum Harmony operators that describe a language.  This also provides a new algorithm for training quantum Boltzmann machines that requires no approximations and works in the presence of hidden units.  We additionally show that quantum language processing is \BQP-complete, meaning that it is polynomially equivalent to the circuit model of quantum computing which implies that quantum language models are richer than classical models unless $\BPP=\BQP$.  It also implies that, under certain circumstances, quantum Boltzmann machines are more expressive than classical Boltzmann machines.  Finally, we examine the performance of our approach numerically for the case of classical harmonic grammars and find that the method is capable of rapidly parsing even non-trivial grammars.  This suggests that the work may have value as a quantum inspired algorithm beyond its initial motivation as a new quantum algorithm.
\end{abstract}
\maketitle

\section{Introduction}

This paper develops an approach to natural language processing for quantum computing. The approach is based in artificial neural networks, because like quantum computers, neural network computers are dynamical systems with state spaces that are high-dimensional vector spaces. The method proposed here follows a general neural network framework for artificial intelligence and cognitive science called Gradient Symbolic Computation (GSC) \cite{smolensky2006harmonic,smolensky2014optimization}. Since GSC takes its starting point from quantum mechanics, this work amounts to the closing of a conceptual circle.

Quantum computation has in recent years been applied to address a host of problems in cryptography~\cite{shor1999polynomial}, simulation of physical systems~\cite{lloyd1996universal,whitfield2011simulation,jordan2012quantum} and machine learning~\cite{aimeur2006machine,rebentrost2014quantum,wiebe2016quantum,kerenidis2017quantum,gilyen2017optimizing}.  The advantages of these methods stem from a number of different sources, including a quantum computer's ability to manipulate exponentially large state vectors efficiently and manipulate quantum interference to improve on statistical sampling techniques.  While techniques such as quantum gradient descent~\cite{kerenidis2017quantum,gilyen2017optimizing} and amplitude amplification~\cite{brassard2002quantum} could be used to provide advantages to performing gradient symbolic computation for language processing, as yet this application remains underdeveloped and furthermore the challenges of preparing the necessary states on a quantum computer makes direct applications of these techniques challenging.  For this reason, new representations for language would be highly desirable for applications in language processing.

The paper addresses two fundamental aspects of language processing: the generation of grammatical symbol sequences (along with their constituent structure parse trees), and the determination of the grammaticality of a given symbol sequence, given a grammar. In the paper, after the relevant aspects of Gradient Symbolic Computation are summarized in Section~\ref{sec:TPRs}, a representational schema for encoding parse trees in a quantum computer is proposed in Section~\ref{sec:Fock}, which identifies a connection between language processing and quantum error correction. In Section~\ref{sec:HOps}, the Hamiltonian of the proposed quantum computer---a type of Boltzmann machine---is related to the grammar that it processes. Then Section~\ref{sec:Learning} takes up the problem of learning the parameters of a quantum computer that processes according to an unknown grammar. Both the unsupervised and supervised learning problems are treated, and the complexity of the proposed learning algorithms are presented. Section~\ref{sec:HMaxNumerics} presents numerical simulations of the generation of sentences in formal languages, which are specified by a given set of symbol-rewriting rules. This amounts to an optimization problem, because in Gradient Symbolic Computation, the grammatical sentences are those that maximize a well-formedness measure called Harmony. Harmony values are physically realized as expectation values of the negative Hamiltonian of the quantum computer.

It should be emphasized that the analyses of supervised learning presented here (in particular, the  computation of the gradient in Theorem~\ref{thm:deriv} and the complexity result in Theorem~\ref{thm:Complexity}) are not restricted to language processing: they apply to supervised training of any quantum Boltzmann machine.

\section{Tensor Product Representations}
\label{sec:TPRs}

	The core of the neural network framework deployed here, Gradient Symbolic Computation (GSC), is a general technique called Tensor Product Representation (TPR) for embedding complex symbol structures in vector spaces \cite{smolensky1990tensor}. For our language applications, the relevant type of symbol structure is a binary parse tree, a structure that makes explicit the grouping of words into small phrases, the grouping of smaller phrases into larger phrases, and so on recursively up to the level of complete sentences, as in $[_{\rm S}~ [_{\rm NP}~ this] [_{\rm VP}~ [_{\rm V}~ is] [_{\rm NP}~ [_{\rm Det}~ an] [_{\rm AP}~ [_{\rm A}~ English] [_{\rm N}~ sentence]]]]]$ \cite{jurafsky2014speech}.  This bracketed string denotes the binary tree shown in Figure \ref{fig:tree}. Each labeled node in the tree is a constituent.

	In one type of TPR embedding---which uses `positional roles'---the vector that embeds a symbol structure ($\mathbf{S}$) is the superposition of vectors embedding all the structure's constituents, and the vector embedding a constituent---a tree node labeled with a symbol---is the tensor product of a vector embedding the symbol ($\mathbf{s}_{i}$) and a vector embedding the position of the node within the tree ($\mathbf{n}_{i}$):
	$\mathbf{S} = \sum_{i} \mathbf{s}_{i} \otimes \mathbf{n}_{i}$
	A position in a binary tree can be identified with a bit string, such that 011 denotes the left (0) child of the right (1) child of the right child of the tree root. (An `only child' is arbitrarily treated as a left child, and the root is identified with the empty string $\varepsilon$.) Thus in this positional-role TPR, the vector that embeds the parse tree for \emph{this is an English sentence} is, in Dirac notation:
\begin{equation}
\label{eq:tree}
\ket{\psi} = \ket{{\rm S}} \ket{\varepsilon} + \ket{{\rm NP}}\ket{0} + \ket{\rm this}\ket{00} + \ket{{\rm VP}} \ket{1}+\ket{{\rm V}} \ket{01}+\ket{{\rm  is}}\ket{001}+\ket{{\rm NP}}\ket{11}+ \cdots
\end{equation}

\begin{figure}
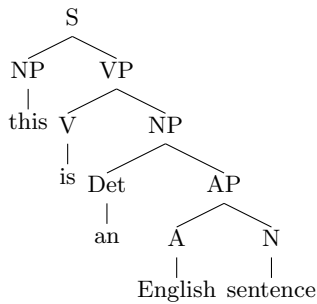

\hspace{-2in}
\Tree [.S [.NP this ] [.VP [.V is ] [.NP [.Det an ] [.AP [.A English ] [.N sentence ] ] ] ] ] 
\caption{A parse tree for \emph{this is an English sentence}. S = sentence, N = noun, V = verb, Det = determiner, A = adjective, P = phrase.}
\label{fig:tree}
\end{figure}

The vectors $\{\ket{\rm S}, \ket{\rm NP}, \ket{\rm this}, \ldots\}$ lie in a vector space $V_S$  hosting the embedded symbols, while the vectors  $\{\ket{\epsilon},\ket{0},\ket{1},\ket{01},\ldots\}$ lie in a vector space $V_N$ hosting the embedded nodes; $\ket{\psi}$ then lies in the tree-embedding space $V_T := V_S \otimes V_N$. Letting the embedding of nodes be recursive, we have $\ket{101} = \ket{0}\ket{1}\ket{1}$. Thus, letting $V_0$ be the vector space spanned by $\{\ket{0}, \ket{1}\}$, we have $\ket{101} \in V_0 \otimes V_0 \otimes V_0 = V_0^{\otimes 3}$. Similarly we have that the vector space of embeddings of all nodes, $V_N$, is the direct sum of the vector spaces containing the  vectors embedding nodes of all depths:
\begin{equation}
V_N = \bigoplus_{d=0}^\infty V_0^{\otimes d},
\end{equation}
as in a multi-particle state space where $V_0^{\otimes d}$ is the space of $d$ particles (and $V_0$ is the single-particle state space). Here, $d$ is the depth of a node in a tree.

	It has been shown that using such TPRs and purely neural network computation, it is possible to compute families of recursive symbolic functions mapping binary trees to binary trees that are relevant to language processing \cite{smolensky2012symbolic}. (That is, for such a function $f$, a neural network can map the embedding of tree $T$ to the embedding of tree $f(T)$.) For the family of functions that are the closure of the primitive tree-manipulating operations---extract left/right subtree, merge two subtrees into a single tree---linear neural networks suffice: such a function can be computed by a single matrix-multiplication \cite{smolensky2006harmonic}.
	
	The general state in $V_{T}$ is not the embedding of a single tree but rather the weighted superposition of embeddings of trees. Thus if $\ket{\psi}$ is the embedding of the parse tree of \emph{this is an English sentence} given in Equation \ref{eq:tree}, and $\ket{\phi}$ is the corresponding parse tree of \emph{this is an American sentence}, then one state in $V_{T}$ is $\ket{\chi} = \frac{1}{2} (\ket{\psi} + \ket{\phi})$. $\ket{\chi}$ embeds the Gradient Symbol Structure which is the parse tree of \emph{this is an }$\frac{1}{2}$(\emph{English }+ \emph{American})\emph{ sentence}; here node 00111 is labeled by the blend of two Gradient Symbols: $\frac{1}{2}$English and $\frac{1}{2}$American. In a general Gradient Symbolic Structure, nodes are labeled by linear combinations of symbols.
	
	In Gradient Symbolic Computation, a grammar is a function that measures the degree of well-formedness of a state in a neural network; this is a Lyapunov function. In a Hopfield net, the network dynamics minimizes a function called the `energy' \cite{hopfield1982neural}; in GSC, the network dynamics maximizes a function $H$ called the `Harmony' \cite{smolensky1986information}). 
 The connection between the network well-formedness function $H$ and grammar derives from $H$ being a linear combination of grammatical constraint functions $f_C$ each of which measures the degree to which a state violates a grammatical requirement: such an $H$ is called a Harmonic Grammar \cite{legendre1990harmonic}. Thus given the constraint $C_1 :=$ `a sentence has a subject' \cite{GrimshawSamek98}, the vector $\ket{\xi}$ embedding the parse tree of \emph{is an English sentence} violates $C_1$ once, hence $f_{C_1}(\ket{\xi}) = 1$. The coefficient of $f_{C_1}$ in $H$, $w_{C_1}$, is a negative quantity so the missing subject lowers Harmony by $|w_{C_1}|$; this is the strength of the constraint $C_{1}$ in the Harmonic Grammar $H$. Harmonic Grammars have proved to be valuable in analyzing natural languages \cite{pater2009weighted}\footnote{{Especially valuable are the special Harmonic Grammars in which each $w_{C_{k}}$ exceeds the maximal possible Harmony penalty arising from the linear combination of all the constraints weaker than $C_{k}$: these are the grammars of Optimality Theory \cite{prince1993optimality, prince1997optimality}. In such a grammar, numerical weighting is no longer required: only the ranking of constraints from strongest to weakest matters for computing which of two structures has higher Harmony. A structure is optimal---grammatical---iff no other structure has higher Harmony.}}.

	The well-formed---i.e., grammatical---sentences are those with globally-maximal Harmony. In a neural network, these can be computed via simulated annealing, in which the stochastic network state follows a Boltzmann distribution $p_T(x) \propto e^{H(x)/T}$; during computation, $T \rightarrow 0$. (Such networks are Boltzmann Machines \cite{ackley1985learning} or Harmony Networks \cite{cho2017incremental}.) Such Boltzmann distributions also describe the states of interest in the quantum analog of the Harmonic Grammar $H$.
	
	As well as describing natural languages, Harmonic Grammars can also describe formal languages \cite{smolensky1993harmonic}, that is, sets of symbol sequences derived by repeated application of rewrite rules such as $\cG_1 := \{S \rightarrow a S b, S \rightarrow a.b\}$, which generates the formal language $\cL_1 := \{a^n.b^n | n = 1, 2, \ldots\}$ (with parse trees $T_1 := \{[_S~a~[_S~a \cdots [_S~a~.~b]~b] \cdots b] \}$). The Harmonic Grammar $H_{\cG_1}$ corresponding to $\cG_1$ will be defined precisely below, but briefly, $H_{\cG_1}$ assigns negative Harmony to all ungrammatical sequences of $a$s and $b$s, and assigns maximal Harmony---$H = 0$---to all grammatical parse trees, $T_1$. One set of constraints assigns negative Harmony to the presence of $a$ or $b$, while another set of constraints assigns positive Harmony to tree configurations that accord with the grammar rewrite rules, e.g., a mother/left-daughter pair in which the mother node is labeled $S$ and the daughter node $a$. The contributions of the negative-Harmony constraints are cancelled by the contributions of the positive-Harmony constraints only for grammatical trees.
	
	TPRs provide a highly general framework for embedding in vector spaces symbol structures of virtually any type, not just trees. In general, a type of symbol structure is characterized by a set of \emph{roles}, each of which can be bound to \emph{fillers}. In the TPR scheme discussed so far, the fillers are symbols and the roles are the tree nodes, that is the positions in the tree. This is an instance of \emph{positional roles}. There is another mode of deploying TPRs which becomes of interest for quantum computing: this method deploys `contextual roles'. Rather than characterizing the role of a symbol in a structure by the position it occupies, we characterize the role of a symbol by its context of---say, $\gamma = 2$---surrounding symbols. For the sequence $abcd$, rather than identifying $b$ as the symbol in second position, we now identify it as the symbol preceded by $a$ and followed by $c$. So the embedding of $abcd$ is no longer $\ket{a}\ket{1} + \ket{b}\ket{2} + \ket{c}\ket{3} + \ket{d}\ket{4}$ but rather $\ket{abc} + \ket{bcd} = \ket{a}\ket{b}\ket{c} + \ket{b}\ket{c}\ket{d}$. For neural networks this contextual-role scheme quickly becomes prohibitive when the number of symbols is large (such as the number of English words).
	
	In the limit, the context size $\gamma$ used to characterize the role of a symbol is large enough to encompass the entire structure. In this limit, for a binary tree, the tree positions are enumerated---$ (p_k)_{k=1,2,\ldots} := (\epsilon, 0, 1, 00, 01, 10, 11, \ldots)$---and then a tree with symbols $(s_k)$ in positions $(p_k)$ is embedded as the vector $\ket{s_1s_2s_3\ldots}=\ket{s_1}\ket{s_2}\ket{s_3}\cdots$. It is such a maximal-contextual-role TPR we deploy below for quantum computation.
	
	Although the state space required is typically considerably larger for contextual than for positional roles, the contextual role scheme has a significant advantage over the positional scheme: superposition can preserve identity. If we superimpose $\ket{abc}$ and $\ket{xyz}$, in the positional scheme we lose the identity of the two sequences, since then $\ket{abc} + \ket{xyz} = (\ket{a} + \ket{x})\ket{1} +(\ket{b} + \ket{y})\ket{2} + (\ket{c} + \ket{z})\ket{3} = \ket{ayz} + \ket{xbc} = \ket{xbz} + \ket{ayc} = \cdots$. But in the contextual scheme, $\ket{abc} + \ket{xyz} = \ket{a}\ket{b}\ket{c}+ \ket{x}\ket{y}\ket{z}$, which is unambiguous.

\section{Fock Space Representations}
\label{sec:Fock}

Despite the presence of the tensor product structure exploited by positional-role tensor product representations (pTPRs) for language, implementing them directly on quantum computers can be a challenge.  This is because the natural representation of a pTPR would be as a quantum state vector.  While such a quantum state vector could be expressed using a very small number of quantum bits, the manipulations needed to manipulate these state vectors to maximize Harmony are non-linear.  Since quantum computers cannot deterministically apply non-linear transformations on the state, this optimization involves non-deterministic operations that can require prohibitive amounts of post-selection in order to apply.

For this reason, we propose using TPRs deploying maximal contextual roles for encoding language structures in a quantum computer.  This will be called a Fock-space representation.  The idea behind the Fock space is that we consider each role that could be filled within the representation as a tensor factor within structures built from a decomposition of  this space.  This is different from pTPR structures wherein linear combinations of tensor products are used to represent symbol structures.  Here every possible combination of roles and fillers are described using a tensor product of simpler terms.  For example, if there are $R$ roles then the basis for this Fock space can be expressed
in Dirac notation as:
\begin{equation}
\ket{v}=\ket{f_1}\cdots \ket{ f_R},
\end{equation}
where each $f_{i}$ is the filler (symbol) bound to positional role $r_{i}$, an ordering of the roles $r_{1}, \ldots, r_{R}$ having been imposed. Thus $\ket{abc} = \ket{a}\ket{b}\ket{c}$.

There are many ways that one could define the basis.  The following convention is proposed here.  Let $\ket{0}$ represent a positional role that does not have a filler, or equivalently let $\ket{0}$ represent a positional role that is filled with the empty symbol ``$0$''.  Next, let $a^\dagger_{f,r} \ket{0}$ be the vector that stores the filler $f$ in the positional role $r$.  This means that any basis vector in the space of fillers on $R$ roles can be written as
\begin{equation}
\ket{v} = a^\dagger_{f_1}\ket{0}\otimes \cdots \otimes a^\dagger_{f_R} \ket{0}:=a^\dagger_{f_1,r_{1}}\cdots a^\dagger_{f_R,r_{R}} \ket{0}
\end{equation}
Here each $a^\dagger_{f_i,r_{j}}$ is a binding operator, which acts in exact analogy to the creation operator in quantum mechanics.  Similarly, define each $a_{f_i,r_{j}}$ to be the corresponding unbinding operator, which maps a bound role back to an unbound role and is the Hermitian transpose of $a^\dagger_{f_1,r_{1}}$.  The properties of binding and unbinding operators are summarized below.
$$
\begin{array}{|c|c|}
\hline
\textbf{Properties of binding operators} &\\
\hline
\textbf{Linearity}& a^\dagger_{f,r} (\alpha\ket{\psi} +\beta\ket{\phi}) = \alpha a^\dagger_{f,r}\ket{\psi} +\beta a^\dagger_{f,r}\ket{\phi},~\forall~f,r,\ket{\psi}, \ket{\phi} \text{ and }\alpha,\beta \in \mathbb{C}\\
\textbf{Distributivity} & a^\dagger_{f,r}( a^\dagger_{f',r'} + a^\dagger_{f'',r''}) = a^{\dagger}_{f,r}a^\dagger_{f',r'} + a^\dagger_{f,r}a^\dagger_{f'',r''}~\forall~r,f,r',f',r'',f''\\
\textbf{Unique Binding (Nilpotence)}& a^\dagger_{f,r}a^\dagger_{f',r} =0=a_{f,r}a_{f',r}~\forall~f,f',r\\
\textbf{Zero Expectation}& \bra{0}a^\dagger_{f,r}\ket{0}=0=\bra{0}a_{f,r}\ket{0}~\forall~f,r\\
\textbf{Number Operator} & \text{ For $n_{f,r} := a^\dagger_{f,r} a_{f,r}$, } n_{f,r} a^\dagger_{f,r} \ket{0} = a^\dagger_{f,r} \ket{0}\text{ and } n_{f,r} \ket{0}=0.\\
\hline
\hline
\textbf{ Classical binding operators} &\\
\hline
\textbf{Commutativity} & a^\dagger_{f,r} a^\dagger_{f',r'} = a^\dagger_{f',r'}a^\dagger_{f,r}~\forall~r,f,r',f'.\\
\hline
\end{array}
$$

{Although Fock space representations need to use classical binding operators, in fact
classical binding operators are special because Fock space representations in general require non-classical (or quantum) binding operators.}  We show this formally in the theorem below.

\begin{theorem}
Classical Fock space representations are generalizations of pTPRs in the following sense: for any pTPR encoding that uses a finite set of orthonormal role and filler vectors to encode structures in which a unique filler is assigned to each role, there exists an injective map from the space of pTPRs to Fock space representations, but there does not exist a bijective map.
\end{theorem}
\begin{proof}
Without loss of generality, let us assume that there does not exist any recursive structure in the pTPR.  This can be done because for any component of the form $[\mathbf{A}\otimes \mathbf{r}_0  + \mathbf{B}\otimes \mathbf{r}_1]\otimes \mathbf{r}_2 \equiv \mathbf{A}\otimes \mathbf{r'}_0 + \mathbf{B}\otimes \mathbf{r'}_1$ by expanding the tensor products and redefining the roles.

Since there are a finite number of roles and fillers in a TPR, for concreteness let us assume that there exist $N$ possible fillers $\{\mathbf{A}_j: j=1,\ldots, N\}$ and $M$ possible roles $\{\mathbf{r}_{k}:k=1,\ldots, M\}$.  Similarly, let $s:\{1,\ldots, M\} \rightarrow \{1,\ldots,N\}$ be a sequence of fillers that are used to represent a fixed but arbitrary pTPR such as
\begin{equation}
\textbf{v}_{\rm pTPR} = \sum_{j=1}^M \mathbf{A}_{s(j)} \otimes \mathbf{r}_j.
\end{equation}

Now let us construct an equivalent vector within a Fock space representation.  For each binding $A_{s(j)} \otimes \mathbf{r}_j$ we can associate a classical binding operator $a^\dagger_{s(j),r_j}$ acting on a different tensor factor.  That is to say
\begin{equation}
\textbf{v}_{\rm pTPR} \mapsto a^\dagger_{s(M),r_M}\cdots a^\dagger_{s(1),r_1}\ket{0}.
\end{equation}
Since $a^\dagger_{f,r}\ket{0} = a^\dagger_{f',r'}\ket{0}$ if and only if $f=f'$ and $r=r'$, which follows
from the definition of  the binding operator,
it follows that two different pTPRs are mapped to the same Fock space representation if and only if they are the same vector.  Hence an injection exists between the representations.

A surjection, on the other hand cannot exist.  To see this, let us examine the dimension of the pTPR.  We have assumed that the pTPR exists in a vector space of dimension $MN$, which follows from the unique binding assumption.  On the other hand, the vector space for the Fock representation is of dimension $(N+1)^M$ (the base is $N+1$ rather than $N$ because of the presence of the vacuum symbol $0$).  Since the dimensions of the spaces are different, it is impossible to construct a surjective map from pTPR to Fock space representations, unless further restrictions are made on the vectors permitted by Fock space representations.  This completes our proof that Fock space representations are a generalization of pTPRs.
\end{proof}

At first glance the proof of the above theorem may seem to suggest that Fock space representations are less efficient than pTPRs.  In fact, even though the vector space that the Fock space equivalent of a pTPR lies in is exponentially larger, the memory needed to store the vector representing a given structure is equivalent.  Indeed, the existence of an injective mapping shows that a pTPR can be easily expressed in this form, revealing that there cannot be in principle a difference in the memory required between the two.

However, the exponentially larger space that the Fock-space representations lie in make this a much more convenient language to describe distributions over TPRs or uncertainty in the fillers that are assigned to given roles.  Just as probability distributions over $n$ bits live in $\mathbb{R}^{2^n}$, having the ability to work in an exponentially larger space makes it convenient for expressing uncertainty in the binding assignments.  This property also allows us to represent quantum distributions over bindings, which further makes this representation indispensable when looking at language processing on quantum computers.

\begin{figure}[t!]
\includegraphics[width=0.5\linewidth]{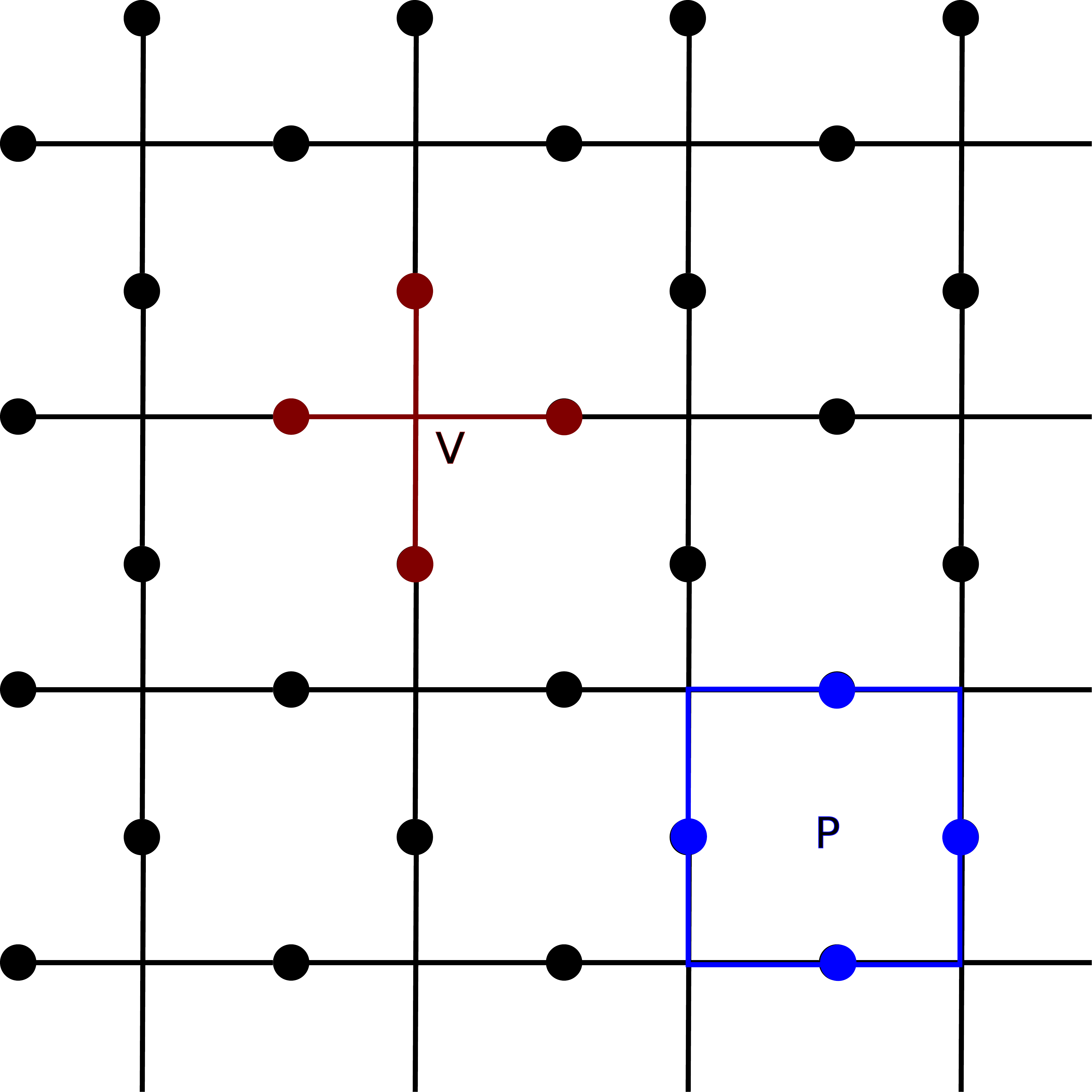}
\caption{Patch of surface code with circles denoting the occupation states for the Fock-space representation (equivalently qubits in the language of quantum computing).  The red cross gives an example of a `vertex' interaction between four occupation states and the blue square is an exemplar of the plaquette interaction between the four occupation states located on the perimeter of the square.  The sets $V$ and $P$ mentioned in the text consist of the union of all such vertex and plaquette subgraphs within the  surface code.}
\label{fig:toric}
\end{figure}

\section{Harmony operators}
\label{sec:HOps}
In Gradient Symbolic Computation, a Harmony function is optimized to determine whether a sentence is grammatical.  The objective within that framework is to find grammatical sentences by \check{globally} optimizing the Harmony, which is a measure of how well-formed a sentence is.  By convention, negative Harmony is associated with ungrammatical sentences and zero Harmony with grammatical sentences.

As an example, consider the following grammar, with fillers $\{S,A,B,.\}$, which generates strings of the form $A^n\ . \ B^n$ for any integer $n$, where $S$ is a start symbol and $\{A, ., B\}$ are terminal symbols.  We can represent this by building a ternary tree that takes the form of a herring bone and assigning roles $\{c_0,l_1,c_1,r_1,l_2,c_2,r_2,\ldots\}$ to the fillers (where $l, c, s$ denote `left, center, right' daughter nodes).  The simplest such tree generated by the grammar takes the form in a pTPR of $S\otimes c_0 + A\otimes l_1 + .\otimes c_1 + B\otimes r_1$.  A choice of Harmony function that works for this assigns Harmony $-3$ to $S$ in role $c_0$, Harmony $-4$ to $S$ placed in $c_d,~d>1$,  and a Harmony penalty of $-1$ for all other symbols.  Harmony bonuses of $+2$ are given if both $S\otimes c_{d-1}$ and any of $A\otimes l_d$, $S\otimes c_{d}$, $.\otimes c_d$, or $B \otimes r_d$ are  bound.  The Harmony of such a tree is then $0$ and thus it is grammatical.  The same rules can easily be generalized to arbitrarily long examples of this grammar.

In Fock space representations we also have the notion of Harmony but the concept of Harmony  needs to be more general in this framework.  This stems from the fact that for Fock space representations the natural generalization of Harmony is an operator rather than a function, as seen below.
\begin{definition}
A Harmony operator $\cH$ for a Fock space representation of $N$ dimensions  is a Hermitian matrix in $\mathbb{C}^{N\times N}$ and the grammatical sentences then correspond to principal eigenvectors of $\cH$.
\end{definition}

As a particular example of such a Harmony operator, let us consider the previously discussed  $A^n.B^n$ grammar.  The Harmony operator for an arbitrary depth sentence can be expressed as
\begin{equation}
\mathcal{H} = n_{c_0,S} - \sum_{j=0}^\infty\sum_{r\in \{r_j,c_j,l_j\}}(4n_{S,r} + n_{A,r} + n_{B,r}  +n_{.,r}) + 2\sum_{j=0}^\infty \sum_{r\in \{r_j,c_j,l_j\}} n_{S,c_{j-1}}(n_{S,c_j}+n_{A,l_j}+n_{.,c_j}+ n_{B,r_j}).\label{eq:treeExample}
\end{equation}
Note that in this particular context, the Harmony operator can be thought of as a function rather than an operator because the Harmony operator is a sum of number operators which can each be represented as a diagonal matrix.  Thus the Harmony operator can be replaced (at a conceptual level) by a function that yields the Harmonies for each possible configuration of the system.

This need not be the case in general.  \check{Settings are possible in which non-classical effects appear.\footnote{\check{Under the most quantum-mechanical interpretation, where the phase of isolated states has no observable consequence, $-\ket{\psi}$ and $\ket{\psi}$ have the same interpretation.
Then
$\ket{\psi_{+}} := (\ket{\rm J} + \ket{\rm K}) \ket{\rm Subject} + \ket{\rm left}\ket{Verb}$
and
$\ket{\psi_{-}} := (\ket{\rm J} - \ket{\rm K}) \ket{\rm Subject} + \ket{\rm left}\ket{Verb}$
also have the same interpretation, an ambiguous blend of the interpretations `Jay left' and `Kay left'.
However
$\ket{\psi} := \frac{1}{2}(\ket{\psi_{+}} + \ket{\psi_{-}})$
is unambiguously `Jay left'.
That the superposition of two ambiguous states can be unambiguous is a purely quantum effect.
}}}
For this reason we introduce the following dichotomy between Harmony operators:
\begin{definition}
A Harmony operator $\mathcal{H}$ is classical if, for all number operators $n_{f,r}$ used in the language, $\cH$ satisfies $n_{f,r}\cH=\cH n_{f,r}$.
\end{definition}

A natural example of a classical Harmony operator is given in~\eq{treeExample}.  It is clearly a classical Harmony operator because it only depends on number operators, which do not change the fillers bound to any role.  Thus the order in which you (i) count whether a role holds a particular filler and (ii) apply the Harmony operator does not matter, and hence the example is classical.

As an example of a non-classical Harmony operator consider the following.  Assume given a square lattice where each vertex in the graph holds a role (see Figure \ref{fig:toric}).  Let $P$ be the set of plaquettes in the graph (meaning the set consisting of all the unit cells in the graph which each consist of 4 vertices because the graph is square) and let $V$ be the vertex set for the graph (meaning the set of all sets of $4$ vertices about each vertex in the square graph).  The language in this case is generated by a single filler and the Harmony operator can be expressed as
\begin{align}
\mathcal{H}=&-\sum_{\{r_1,r_2,r_3,r_4\} \in P} (1 - 2n_{r_1})\otimes(1 - 2n_{r_2})\otimes (1 - 2n_{r_3})\otimes (1 - 2n_{r_4})\nonumber\\
&-\sum_{\{r_1,r_2,r_3,r_4\} \in V}(a_{r_1}^\dagger + a_{r_1})\otimes(a_{r_2}^\dagger + a_{r_2})\otimes (a_{r_3}^\dagger + a_{r_3})\otimes (a_{r_4}^\dagger + a_{r_4}).
\end{align}
This Harmony operator corresponds to the Toric code, which can be used as an error correcting code for quantum computing.  The states of maximum Harmony correspond to the minimum energy subspace of the code, which is proven to be protected from local error.  This shows that apart from mere academic curiosity, non-classical Harmony operators are vitally important for quantum computing and also that quantum error correction has deep links to linguistics when viewed through this lens.

The expression for the Harmony of a given Fock space representation takes the same form regardless of whether we have a classical or a quantum Harmony operator.
\begin{definition}
Let $\ket{\phi}\in \mathbb{C}^{N}$ be a Fock space representation of a sentence and let $\mathcal{H}\in \mathbb{C}^{N\times N}$ be a Harmony operator.  We then define the Harmony of $\ket{\phi}$ to be $H(\ket{\phi}) = \bra{\phi} \mathcal{H} \ket{\phi}$.
\end{definition}

Finding a grammatical sentence---one that is maximally Harmonic---then boils down to optimizing the expectation value of the Harmony operator. However, this optimization inherently incurs a cost.  We assess the cost of both quantum and classical Harmony optimization using an oracle query model.  Within this model we assume that nothing is known about the Hamiltonian, save what can be gleaned from querying the oracle that represents the Harmony function.

\section{Harmony maximization: numerical simulation}
\label{sec:HMaxNumerics}

From the perspective of gradient symbolic computation, the goal of parsing a sentence within a grammar is to find the assignment of roles and fillers that maximize the Harmony.  Here we will look at the problem of optimizing Harmony for classical Fock spaces, which is to say where the Harmony operator is just a sum of number operators.  We will see from these examples that optimizing Harmony within a Fock space representation is practical and as such providing quantum speedups to the learning process is significant.

\subsection{The $A^n \, . \, B^n$ grammar}

Recall that the parse tree rules for grammatical expressions of the form $A^n \, . \, B^n$ are defined over the four-symbol alphabet $\{A,B,S,.\}$.
This is a simple example, where the parse tree can be visualized as a ``herring bone'' structure (HB) that can be recursively described as follows:

0) zero depth HB consists of one node;
1) the root of an HB of depth $n$ has exactly three children: the children number 1 and 3 are leaves and child number 2 is an HB of depth $n-1$.

When the $n$ is chosen, the corresponding Fock space and the Harmony operator are fully defined, and Harmony being a diagonal operator, it can be reinterpreted as a certain scalar function $h$ on the space of all possible assignments of symbols to the nodes of the HB structure.

Negative harmony $-h$ can be thus treated as a Hamiltonian of the corresponding Potts model, which is a generalization of the Ising model \cite{Gallavotti1999}, \cite{potts1952} on the HB graph. The difference between the model at hand and the traditional Ising model is that the Ising model consists of 2-value spins, whereas in our instance of the Potts model each node can assume one of four values in $\{A,B,S,.\}$. The maximum Harmony assignment of these values is understood as a ground state of the Hamiltonian $-h$.

In this setting we can find such ground state by the use of the simulated annealing strategy that has an excellent track record in solving Ising models. An outline of an algorithm for solving a more general Potts model is as follows:

\begin{algorithm}[H]
\caption{Simulated annealing for Potts model.}
\label{alg:potts:annealing}
\algsetup{indent=2em}
\begin{algorithmic}[1]
\REQUIRE{Coupling graph $G$, initial symbol assignment $A_0$ to nodes of $G$, maximum iterations $maxUp$; hyperparameter: cooling schedule $t(i),i=1,\ldots,N$}
\ENSURE{initial symbol assignment $A_{opt}$}
\STATE {$h \gets Harmony(A_0)$; $A \gets A_0$; $A_{opt} \gets A$}
\FOR {$i \in \{1,\ldots,N\}$}
\STATE {$\beta \gets 1/t(i)$}
\FOR {$c \in \{1,\ldots,maxUp\}$}
\STATE {$u \gets$ random symbol update; $h' \gets Harmony(u(A))$}
\IF {($h' \geq h) || (rand() < exp(\beta (h'-h))$ }
\STATE {$A \gets u(A)$; $h \gets h'$}
\IF {$h=0$}
\RETURN{A}
\COMMENT{Early breakout on perfect Harmony}
\ENDIF
\ENDIF
\ENDFOR
\ENDFOR
\end{algorithmic}
\end{algorithm}

For any pre-selected $n$ there is a unique assignment of symbols that turns the HB structure of depth $n$ into a zero Harmony parse tree.
Our experiments indicate that this unique grammatical HB structure can be attained by the algorithm \ref{alg:potts:annealing} that starts from a random symbol assignment in $O(n)$ steps on average.
The numeric tests used highly optimized simulated annealing code modified to accomodate Potts models. We tested HB structures with $n \in [2..1024]$ measuring the minimal number of repetitions and sweeps of the annealing process required for achieving the maximum Harmony. This goal was consistently achievable with 10 repetitions and 20 sweeps independently of $n$. The minimal number of sweeps would occasionally fall to 19 in about $10\%$ of cases and it was registered at 18 in just one case. Since the structure of depth $n$ has $3\,n+1$ nodes, one can say that the maximization required roughly $600 n$ reevaluations of the Harmony function on average with at most $10\%$ variance. The empirical average complexity of harmonizing to the $A^n . B^n$ expression (as a function of $n$) is shown on Fig. \ref{Fig:annealing:AnBn}.

\begin{figure}[t]
\centering
\includegraphics[width = 0.9\textwidth]{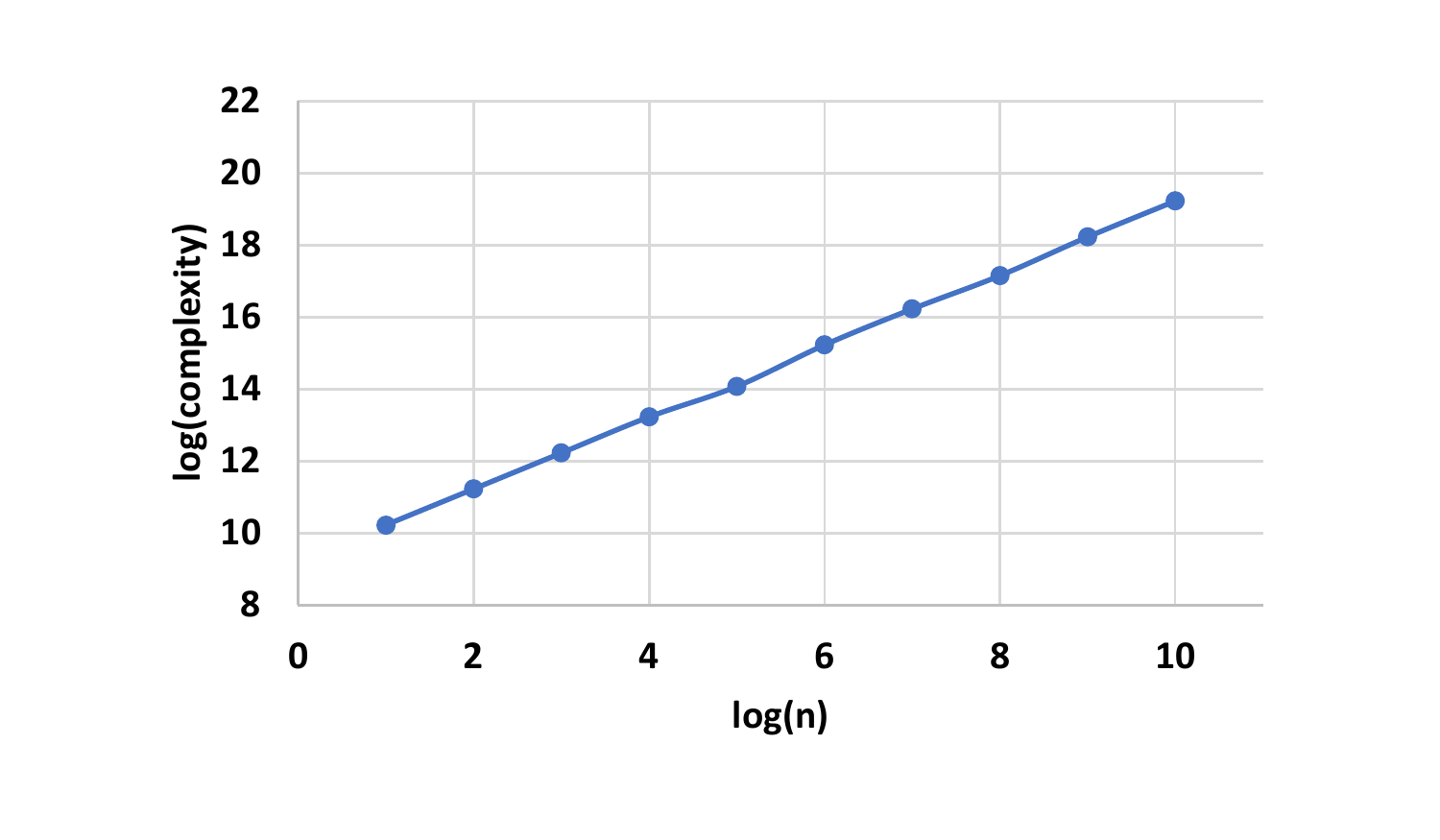}
\caption{Annealing complexity as measured by the number of times the Harmony function must be queried to find the tree of maximum Harmony for the $A^n . B^n$ grammar. (Dual log scale.)}
\label{Fig:annealing:AnBn}
\end{figure}

\subsection{The balanced parentheses grammar}

The balanced parentheses grammar to enumerate and error correct grammatical expressions composed of left and right parentheses is a grammar over the alphabet of 6 symbols A,B,C,S,(,) and the following set of normalized generative rules:

$S\rightarrow B, S \rightarrow C, B \rightarrow ( \, A, B \rightarrow ( \, ), A \rightarrow S \, ), C \rightarrow S \, S$.

Semantically the S symbol can only occur at root of a grammatical parse subtree tree or a complete grammatical parse tree.
For example, Fig \ref{Fig:parse_tree} shows the unique parse tree for the expression ( )( ) that is a concatenation of two disjoint grammatical subexpressions

\begin{figure}[h!]
\centering
\begin{tikzpicture}[
    every node/.style = {shape=rectangle, rounded corners,
    draw, align=center,
    top color=white, bottom color=blue!20},
    level 1/.style = {sibling distance=3cm},
    level 2/.style = {sibling distance=4cm},
    level 3/.style = {sibling distance=2cm},
    level distance = 1.25cm
    ]]
  \node {S}
    child { node{C}
        child { node {S}
            child { node {B}
                child { node {$($}}
                child { node {$)$}}
            }
        }
        child { node {S}
            child { node {B}
                child { node {$($}}
                child { node {$)$}}
            }
        }
         };
\end{tikzpicture}
\caption{An example of optimal Harmony parse tree.}
\label{Fig:parse_tree}
\end{figure}
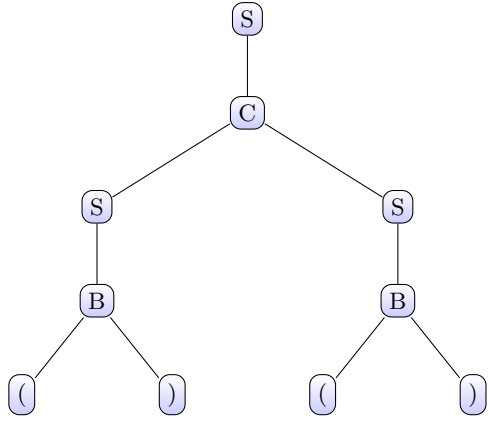

The Harmony function for this grammar is the following.
\begin{table}
\begin{centering}
{\LARGE Harmony Function}
\end{centering}\\
\begin{tabular}{|c|c|}
\hline
Symbol& Harmony\\
\hline
S&-2\\
A &-3\\
B&-3\\
C&-3\\
(&-1\\
)&-1\\
\hline
\end{tabular}
\hspace{5mm}
\begin{tabular}{|c|c|c|}
\hline
Parent&Left Child& Harmony\\
\hline
S&B&2\\
S&C&2\\
B&(&2\\
A&S&2\\
C&S&2\\
\hline
\end{tabular}
\hspace{5mm}
\begin{tabular}{|c|c|c|}
\hline
Parent&Right Child& Harmony\\
\hline
B&A&2\\
B&)&2\\
A&)&2\\
C&S&2\\
\hline
\end{tabular}\\
~\\
\begin{centering}
{\LARGE Harmony Operator}
\end{centering}
\begin{align*}
\mathcal{H} =& \sum_j\left(n_{S,j}(-2 + \delta_{j,1}) -3(n_{A,j} + n_{B,j}+ n_{C,j}) - n_{(,j} -n_{),j}) \right)\nonumber\\
& +2 \sum_j\left(n_{S,j}n_{B,L(j)} + n_{S,j}n_{C,L(j)} + n_{B,j}n_{(,L(j)}+ n_{A,j}n_{S,f(j)}+ n_{C,j}n_{S,f(j)} \right)\nonumber\\
& + 2\sum_j \left( n_{B,j} n_{A,R(j)} + n_{B,j} n_{),R(j)} + n_{A,j} n_{),R(j)} +n_{C,j} n_{S,R(j)} \right)
\end{align*}
\caption{Table describing the Harmony operator for the balanced parenthesis grammar.  A Harmony bonus of $+1$ is assigned for having $S$ at the root of the tree. Unless othewise stated, the Harmony for a given configuration is zero.  We also give the classical Harmony operator for balanced parenthesis grammar on a Fock space consisting of $2^{D}$ modes where we define for any vertex $j$ $L(j)$ to be the left child of the node and $R(j)$ to be the right child.\label{tab:parenHarm}}
\end{table}

Assuming the structure of a candidate parse  tree  is known, so is the structure of the corresponding Fock space, so is the Harmony operator.
As explained above, in case when the Harmony operator is diagonal, it can be cast as a real-valued function $h$ on the space of all possible $node\rightarrow symbol$ assignments for the given parse tree. An optimum-Harmony assignment of the symbols can be then found as an \emph{argmax} of the function $h$ using a suitable maximization method. We demonstrate below how this can be done with a certain simulated annealing approach.
Unfortunately, not every candidate parse tree  allows symbol assignment that realizes the absolute maximum of Harmony. In fact if we consider a set of binary trees of known maximum depth  $D$ and known maximum leaf count $L$ , then the subset of binary trees that allows grammatical assignment is exponentially small vs. the entire set (w.r.t. $D,L$).
We dub a binary tree that allows such maximum Harmony symbol assignment a \emph{feasible parse tree}. All other trees are dubbed \emph{infeasible}.
It follows that the relatively simple code for maximizing Harmony on a given candidate parse tree should be just a subroutine in a higher level algorithm that enumerates all feasible parse trees, or, for error correction purposes, morphs an infeasible tree into a feasible tree.
As shown below, the higher level of the overall algorithm can be also designed along the lines of simulated annealing over a reasonable update heuristics.

\subsubsection{Recursive enumeration of feasible parse trees.}

We start with a specialized Harmony optimization method that exploits the fact that the parentheses placement grammar is context-free. This method is likely to generalize well to any context-free grammar.
We observe that in this context any subtree of an optimal parse tree is optimal. Let us make a stronger observation for the particular Harmony Hamiltonian proposed in Table~\ref{tab:parenHarm}.

\begin{lemma}
 For the Harmony operator in Table~\ref{tab:parenHarm}: an entire harmonical parse tree has the Harmony of 0; any subtree of such tree has the Harmony of -1.
 \end{lemma}
\begin{proof}
Before proceeding with a recursive proof, we recall that there is a Harmony bonus of +1 for symbol S at the root of the entire tree. Disregarding this bonus we can say that the entire tree and any of its subtrees must have the Harmony of -1 in a harmonical parse tree.
Let us first prove, recursively, that a parse subtree with any assignment of symbols cannot have Harmony greater than -1. Indeed it is obvious for subtrees of depth 0. Assuming it has been proven for subtrees of depth at most d consider a parse subtree of depth d+1. Any child subtree of its roots has the Harmony of at most -1 by the induction hypothesis. The Harmony of the symbol assignment at the root is negative. Unless the edges to the child subtrees correspond to the correct generative rules (and thus incur the Harmony bonus of +2), the overall Harmony is going to be less than -1. So let us exhaust cases where the edges do correspond to generative rules.

Case A,B,C: Root assignment of either  A,B,C carries Harmony penalty of -3. If there is only one child subtree with the Harmony ≤-1 the total subtree Harmony cannot exceed -1+2-3 = -2. If there are two child subtrees, the Harmony cannot exceed -1+(-1)+2+2-3 = -1.
Case S: If there are two child subtrees under the root, each with Harmony ≤-1, we note that at most one edge to one of those subtrees can gain the Harmony bonus of +2 (since S has no generative rules with two children). Thus the overall Harmony cannot exceed -1 +(-1) + 2 + 0 -2 = -2. If there is only one subtree under the root, the overall Harmony still cannot exceed -1 + 2 -2 = -1.
Case “(,)”:  the case when the root assignment is one of the parentheses is obvious.
Let us now prove, by case distinction, that in a parse subtree of Harmony -1, then any child subtree of its root must also have Harmony -1.
Case “(,)”:  If the root assignment is either of the parentheses, any child subtree will contribute at most -1 to the overall Harmony. Thus there must be no child subtrees for the overall tree to have Harmony of	 -1. The claim of the observation is trivially valid.
Case A,B,C: Root assignment of either  A,B,C carries Harmony penalty of -3. If there is only one child subtree with the Harmony $\leq -1$ the total subtree Harmony cannot be $-1$. Therefore there are two child subtrees with the harmonies
$h_1\leq -1,h_2 \leq -1$ and the overall Harmony is at most $h=h_1+h_2+4-3=h_1+h_2+1$. Obviously we must have $h_1=h_2=-1$ for $h=-1$.
Case S: If there are two child subtrees under the root, each with Harmony $\leq -1$, we note that at most one edge to one of those subtrees can gain the Harmony bonus of +2 (since S has no generative rules with two children). Thus the overall Harmony cannot be -1. Therefore there is one child subtree with the Harmony $h_c \leq -1$. And the overall Harmony is at most $h_c$. Thus we must have $h_c=-1$.
\end{proof}

\begin{algorithm}[H]
\caption{Recursive function $enumSubtrees(L,D)$.}
\label{alg:enum:subtrees}
\algsetup{indent=2em}
\begin{algorithmic}[1]
\REQUIRE{$L$ parse tree leaf count, $D$ maximum depth of a parse tree}
\ENSURE{Complete list of Harmony $-1$ parse trees of leaf count $L$ and maximum depth $D$}
\IF {$D=0$}
\IF {$L=1$}
\RETURN {$[ root[( ]; root[)]]$}
\ELSE
\RETURN {[]}
\ENDIF
\ENDIF
\COMMENT{First we enumerate all the trees of depth $D$ with only one child subtree under the root}
\STATE {$ret1 \gets []; list1 \gets enumSubtrees(L,D-1)$}
\FOR {$t \in list1$}
\STATE {$cand \gets root[S],child[t]$}
\IF{$Harmony(cand)=-1$}
\STATE{$ret1 \gets ret1+[cand]$}			
\ENDIF
\ENDFOR
\COMMENT {for two-child root we explore all possible splits of leaf counts between children}
\STATE{$ret2 \gets []$}
\FOR {$\ell \in \{1,\ldots,L-1\}$}
\STATE {$lleft\gets enumSubtrees(\ell,D-1)$}
\STATE {$lright \gets enumSubtrees(L-\ell,D-1)$}
\FOR {$t_1 \in lleft, t_2 \in lright$}
\FOR {$s \in {A,B,C}$}
\STATE {$cand \gets root[s],child[t_1 ],child[t_2]$}
\IF {$Harmony(cand)=-1$}
\STATE{$ret2 \gets ret2+[cand]$}
\ENDIF
\ENDFOR
\ENDFOR
\ENDFOR
\RETURN{$ret1+ret2$}
\end{algorithmic}
\end{algorithm}

This algorithm reads as a very expensive doubly recursive routine as it is written. However in practice it can be made perfectly manageable by caching all the previously computed enumerations in a global cache. This way any recursively requested enumSubtrees(l,d) retrieves the answer immediately from the cache iff it has been ever before computed.

\subsubsection{Annealing into feasible parse trees.}

Algorithm~\ref{alg:enum:subtrees}, developed in the previous subsection, is built upon specific properties of the grammar in question and might not generalize cleanly to other grammars. Let us consider a general situation where, given a binary tree which is a candidate parse tree, it is then relatively easy to find a maximum-Harmony assignment of symbols to the nodes of the tree. In particular,  it is relatively easy to conclude algorithmically, whether the candidate tree is feasible. As per discussion in the beginning of the section, feasible trees are quite rare and the probability that a randomly generated tree is feasible, is exponentially low. We would benefit from a strategy that, given a random tree, can morph the tree after an acceptable number of steps into a feasible tree. Such strategy would have an important error correction aspect, as it would be capable of editing an erroneous parse tree into a correct one at a relatively low cost.  A tree morphing strategy needs to be broken up into a sequence of relatively simple steps to be universal and it is intuitively clear that in general the morphing strategy cannot be greedy, i.e. it is in general not possible to reach a feasible tree by a sequence of steps that monotonously increases maximum Harmony of the consecutive candidate trees. Thus, we are again considering the simulated annealing philosophy at this level. Below we propose one possible design for a tree-morphing algorithm.

\emph{Elementary steps.}
We allow the following elementary operations on binary trees:

	1) Leaf deletion: a leaf of the tree is deleted along with the edge leading to it

	2) Leaf creation: a leaf is added to some node with fewer than 2 children.

Clearly this set of operations is universal. Indeed, any tree can be evolved from a root by a sequence of operations of type 2) and any tree can be reduced to a root by a sequence of operations of type 1). Therefore any tree $T_1$ can be morphed into any other tree $T_2$ by a sequence of operations of type 1) and 2). However we have found that it is beneficial in practice to introduce a redundant elementary operation:

	3) Leaf forking: turn some leaf into an interior node by attaching two new leaves to it.

\emph{Morphing under constraints.}
Just as in the previous subsection, we drive a request for a feasible tree by stipulating its desired leaf count $L$ and maximum depth $D$.
We choose the morphing updates such that the depth of the tree post-update never exceeds $D$ and its leaf count stays very close to $L$. Thus we always prefer a leaf forking or leaf creation at interior node, whenever the leaf count falls below $L$; and we never add a leaf to an existing leaf node or fork a leaf if this leads to a tree of depth greater than $D$.  (There is a theoretical possibility of a deadlock in this strategy, where a leaf cannot be added without increasing the depth of the tree beyond the limit, however, this cannot happen when $D>\log_2(⁡L)$, which is the primary scenario.)
The top-level scheme of an annealing-style tree morphing algorithm is as follows:

\begin{algorithm}[H]
\caption{Parse tree morphing (top level).}
\label{alg:tree:morphing}
\algsetup{indent=2em}
\begin{algorithmic}[1]
\REQUIRE{Initial tree $T_0$, maximum depth $D$, maximum iterations $maxUp$; hyperparameter: cooling schedule $t(i),i=1,\ldots,N$}
\ENSURE{Feasible parse tree of leaf count $L$ or  $L-1$}
\STATE {$h \gets Harmony(T_0)$; $T \gets T_0$}
\IF {$T$ is feasible}
\RETURN {$T$}
\ENDIF
\FOR {$i \in \{1,\ldots,N\}$}
\STATE {$\beta \gets 1/t(i)$}
\FOR {$c \in \{1,\ldots,maxUp\}$}
\IF { $leaf count(T) < L$ }
\STATE {$u \gets$ random additive update for $T$}
\ELSE
\STATE {$u \gets$ random reductive update for $T$}
\ENDIF
\IF {$depth(u(T)) \leq D$}
\STATE {$h' \gets$ maximum Harmony on $u(T)$}
\IF {$(h'\geq h)~||~ (rand()< exp⁡(\beta(h'-h))$ }
\STATE { $h \gets h', T \gets u(T)$ }
\IF {$T$ is feasible}
\RETURN {$T$}
\ENDIF
\ENDIF
\ENDIF
\ENDFOR
\ENDFOR
\end{algorithmic}
\end{algorithm}

Here the “random additive update” 	means forking of a randomly selected leaf or adding a leaf as a second child to an interior node. The “random reductive  update” as per our definition of elementary operations means deleting a leaf. In order to spur the convergence and  eliminate  deadlocks we exclude adding leaf(s) to a site where a leaf has been has been  recently deleted and we exclude deletion of a recently added leaf. The “T is feasible” predicate entails maximizing Harmony over all the assignments of symbols to the nodes of the subtree T. This can be easily done in practice by running a suitable simulated annealing subroutine on the set of all possible symbol assignment configurations. The subtree T is feasible iff the maximum Harmony thus achieved is equal to -1.

\subsubsection{Simulation metrics for annealing into feasible parse trees.}

	After requesting a random binary tree of depth at most 4 with 4 leaves, the initial random tree gets morphed into a feasible 4-terminal parse tree in less than 60 elementary moves of our Algorithm \ref{alg:tree:morphing}. For comparison, it typically takes more than 1000 randomly generated 4-leaf sample trees to get a feasible parse tree candidate.
	After requesting a random binary tree of depth at most 6 with 5 leaves in takes about 450 elementary moves in median case to generate a feasible parse tree candidate with the harmony of $-1$\footnote{Obviously, there are no harmonic trees with an odd number of leaves}.After requesting a random binary tree of depth at most 7 with 6 leaves in takes about 2100 elementary moves in median case to get a feasible parse tree. (A quick reference to these numbers is given in the Table \ref{tree:anneal:cost}.) In comparison, we haven't succeeded in finding a feasible tree among the first $100,000$ randomly generated 6-leaf sample trees.

\begin{table}[h!] \caption{Cost of simulated annealing into feasible parse trees}
\begin{tabular}{c c c}
\hline
Depth & Leaves & Number of moves (median) \\ [0.5ex]
\hline
4 & 4 & 55 \\
6 & 5 & 450 \\
7 & 6 & 2100 \\ [1ex]
\hline
\end{tabular}
\label{tree:anneal:cost}
\end{table}

    The algorithm, however, has significant downsides and requires additional work. First of all it is sensitive to the shape of initial tree candidate. In case of unfavorable initialization it could take up to twice as long to terminate than on the average case. We also have registered a single instance of 6-leaf run where it never converged. (This is why we list the median steps to termination rather than ``average''). Another feature of the termination metric is that it is likely to still be exponential in the requested number of leaves. An open research question is whether this is a classical limitation of the problem or can be improved upon with a better algorithm.

\section{Computational Power of Quantum Language Processing}
We have seen that Classical Fock Space representations are can be used to solve problems in language processing, but an important question remains: ``what quantum advantages can be gleaned from using a quantum Harmony operator?''  We provide evidence for two kinds of advantages.  The first such advantage shows that quantum language processing using a reasonable family of quantum binding operators, cannot be efficiently simulated on a classical computer within arbitrarily small error unless $\BQP=\BPP$.  We demonstrate this by recasting the problem of parsing a quantum language to the problem of performing a quantum computation.  The second such advantage is speedups for optimizing classical Harmony functions using quantum simulated annealing.

The issue of the computational harness of quantum learning task has increasingly come to the fore with a number of high-profile dequantizations, or quantum inspired classical algorithms, of quantum algorithms that would seem at first glance to offer exponential speedups~\cite{gilyen2018quantum,tang2018quantum,tang2018quantum2,wiebe2015quantum,wiebe2015quantum2}.  
This points a spear at the heart of the hope that quantum models for data may be more expressive than classical methods.  Here we address this by showing that there are at least some classes of languages such that evaluating the language is equivalent to quantum computation, thus suggesting that our approach is unlikely to ever be dequantized.  Furthermore, these results trivially show that there exists a class of Boltzmann machines that are universal and cannot be dequantized.

In order to demonstrate that quantum Fock-space representations for language are more powerful than classical representaitons, we need to first define a computational model that uses such representations to solve problems.  We do this below.
\begin{definition}
We define a Harmonic quantum computer to be a model for quantum computing that obeys the following assumptions.
\begin{enumerate}
\item  Let $\mathcal{F}$ be a twice-differentiable map from $[0,1]$ to a Harmony operator acting on $O(n)$ modes such that $\|\mathcal{F}(s)\|$ and $\|\partial_s \mathcal{F}(s)\|$ are in $O({\rm poly}(n))$ for all $s\in [0,1]$.
\item Let $\mathcal{H}(s)$ consist of a sum of terms that are formed from products of at most $\kappa\in O(1)$ binding operators and that the coefficient of each such term be efficiently computable.
\item Each binding operator in the Fock space representation can be represented as an $O({\rm poly}(n))$ sparse row-computable matrix.
\item The state of the quantum computer can be set at any time, at cost $O({\rm poly}(n))$, to $\prod_{j=1}^{O({\rm poly}(n))} a^\dagger_j \ket{0}$.
\item Assume that the user can measure the occupation number for each role/filler combination in the language at unit cost and also measure in the eigenbasis of $\mathcal{F}(s)$ for any $s\in [0,1]$ within error $\epsilon$ and probability of failure at most $1/3$ at cost $O({\rm poly}(n/\epsilon))$.
\end{enumerate}
\label{def:lingComp}
\end{definition}


With this definition in place it is easy to see that such a Harmonic quantum computer differs slightly from the type of problems that we have considered previously.
No notion of Harmony optimization is built into the computer.  Additionally, the computer requires a parameterized family of Harmony operators rather than just one.
The requirement that we use a family of Harmony operators is introduced to deal with the fact that Harmony maximization is absent in this model.  Specifically, we solve the problem of Harmony maximization by choosing a Harmony operator that is easy to solve classically and then slowly transform it to the actual Harmony operator that we want to solve.  This is analogous to adiabatic quantum computing~\cite{farhi2000quantum}.

\begin{theorem}
There exists a Harmonic quantum computer that satisfies Definition~\ref{def:lingComp} with $\kappa=4$ that is polynomially equivalent to the circuit model of quantum computing.
\end{theorem}
\begin{proof}
In order to prove our claim we need to show first that there exists a harmonic quantum computer that can simulate any circuit then we will show that this model can be simulated efficiently by a circuit-based quantum computer.
The forward direction of the claim follows immediately from~\cite{GOSSET}, which shows equivalence between the circuit model of quantum computation and adiabatic quantum computing using an XXZ model on a lattice.
The construction used in the work is the following.

\begin{figure}[t!]
\includegraphics[width=0.9\linewidth]{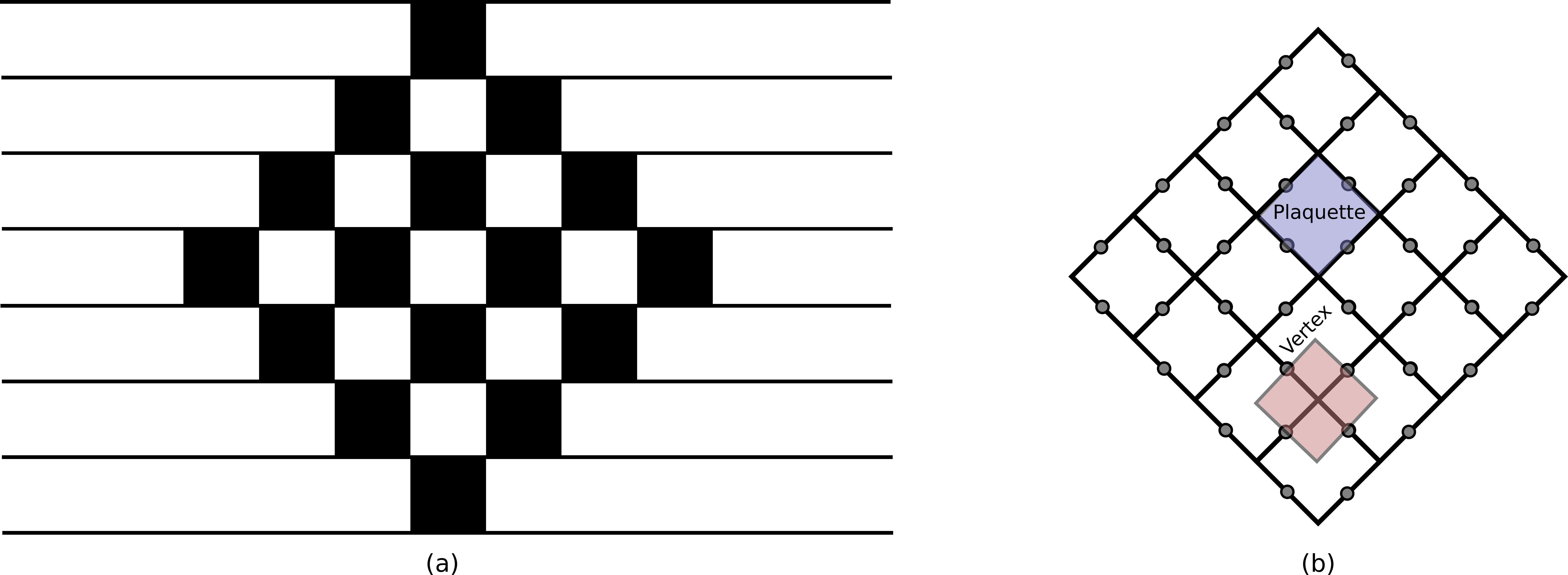}
\caption{Circuit to Hamiltonian construction used in~\cite{GOSSET}.  (a) Represents shows in black boxes a set of two-qubit gates acting on a system of $8$ qubits.  (b) Represents the structure of the Hamiltonian, wherein each circle drawn on an edge represents a site that a particle can be at with examples of the four-body vertex and plaquette operators used in the construction superimposed.}
\end{figure}

Consider a set of two-qubit gates $\{U_p\}$ that act on $2k$ qubits with linear-nearest neighbor connectivity.  While the two-qubit gates are arbitrary, there is an assumed pattern to their targets.  The gates are laid out in a causal diamond.  In the first timestep, gate operations only occur between qubits $k-1$ and $k$.  Similarly in the last timestep there is also only a two-qubit gate between qubits $k-1$ and $k$.  In the second timestep there are two qubit gates between qubits $k-1$ and $k-2$ as well as qubits $k$ and $k+1$.  The same qubits are also targetted by the two-qubit gates in the second-last timestep.  This process is repeated until both patterns intersect at the $k^{\rm th}$ time step.  These two-qubit gates are arbitrary and since they are universal, any quantum circuit can be embedded within this pattern for a sufficiently large $k$, potentially by taking many of the gates to be identity.

A Harmony operator can be constructed that corresponds the these gates.  The Fock space for this Harmony operator consists of roles corresponding to the space-time coordinates that a gate acts within.  For example, let $w$ be a qubit that a particular two-level gate acts on and assume that the gate is active between times $t$ and $t+1$.  The roles correspond to the boundaries of this space-time region: $(w,t), (w+1,t), (w,t+1), (w+1,t+1)$.  The fillers that are placed in each role are $0,1$ which correspond to the values that the qubits that each of the gates act on could take.
We denote the binding operators for the corresponding grammar to be $a^\dagger_{(w,t),f}$ for $f\in \{0,1\}$ and similarly define $n_{(w,t),f} :=a^{\dagger}_{(w,t),f} a_{(w,t),f}$.

If we define $P$ to be the plaquettes formed by the boundaries of the space-time volumes that each of the gates resides within (i.e. the $t$ and $w$ coordinates that bound the space and time that each gate acts within) and let $V$ be the vertex set for the graph then we can then define an indexed family of Harmony operators $\mathcal{F}(s)$ for $s\in [0,1]$ to be~\cite{GOSSET}
\begin{equation}
\mathcal{F}(s) = \sum_{p\in P} \mathcal{H}_{\rm gate}^p(s) +\sqrt{1-s^2} \mathcal{H}_{\rm init} + \sum_{v\in V} \mathcal{H}^{v}_{\rm string} +\mathcal{H}_{\rm input}
\end{equation}
Here we take for convenience $n_{(w,t)} := \left(n_{(w,t),0} + n_{(w,t),1}\right)$ and use $h.c.$ as an abbreviation for Hermitian conjugate
\begin{align}
\mathcal{H}_{\rm gate}^p(s) &:= - \left[n_{(w,t)}n_{(w+1,t)}+n_{(w,t+1)}n_{(w+1,t+1)}+s\mathcal{H}_{\rm prop}^p \right],\nonumber\\
\mathcal{H}_{\rm prop} &:= \sum_{\alpha,\beta,\gamma,\delta} \left(\bra{\beta\delta} U_p \ket{\alpha\gamma} a^\dagger_{(w,t+1),\beta}a_{(w,t),\alpha}a^\dagger_{(w+1,t),\delta}a_{(w+1,t+1),\gamma} \right) +h.c.\nonumber\\
\mathcal{H}^{v}_{\rm string} &:= -\left[n_{(w,t)}+n_{(w,t+1)} +n_{(w+1,t)}+n_{(w+1,t+1)} - 2(n_{(w,t)}+ n_{(w,t+1)})(n_{(w+1,t)}+n_{(w+1,t+1)}) \right],\label{eq:univHarmony}
\end{align}
where we define $\mathcal{H}^v_{\rm string}$ as above for all vertices in the graph that have degree $4$, which is to say that the vertices do not sit at the boundary of the causal diamond.  If the vertex sits at the boundary the terms that couple outside the diamond are set to zero.  We also implicitly assume that the $w's$ indexed in the above terms are the spatial coordinates for the qubits within each plaquette that a given gate $U_p$ acts upon.  Finally, we have that
\begin{align}
\mathcal{H}_{\rm init}&:= -\left[n_{(1,n+1)} + n_{(2n,n+1)} \right],\nonumber\\
\mathcal{H}_{\rm input}&:= -\left[ \sum_{w=1}^{2n} \sum_{t\le n} n_{(w,t),1} \right],
\end{align}
which serves to create a Harmony penalty if the fillers corresponding to the initial qubit state is not set to $0$ at the beginning of the computation (corresponding to $s=0$).
The above Harmony operator satisfies the requirements laid out in Definition~\ref{def:lingComp} with $\kappa=4$ by inspection.

It is further proved in~\cite{GOSSET} that the maximum eigenvalue of the Harmony operator $\mathcal{F}(s)$ is $0$ and the eigenvalue gap for any $s$ is at least
\begin{equation}
\gamma(s) \ge \frac{1}{4n+3}\left(1-s\cos\left(\frac{\pi}{2n} \right)\right)\in \Omega(n^{-3}).\label{eq:gapbd}
\end{equation}

The state of maximum Harmony is shown in~\cite{GOSSET} that if $k= \sqrt{n}/16$ then the configuration with maximum Harmony corresponds can be measured to find the output of the circuit with probability bounded below by a positive constant.  Thus $\mathcal{F}(1)$ is a Harmony operator with $\kappa=4$ whose maximum Harmony configuration yields, after a successful measurement, the result of any quantum computation.

In order to see that the cost is $O({\rm poly}(n))$ for the Harmonic quantum computer note that we have within the model ascribed the cost of measuring the Harmony of the system of the system within error $\epsilon$ and probability at least $1/3$ at cost $O({\rm poly}(n/\epsilon))$.  In order to guarantee that the measurement successfully projects onto the maximum Harmony state, the measurement must have $\epsilon \in O(\gamma)$.  Since $\gamma\in \Omega(n^{-3})$ it suffices to choose $\epsilon \in \Theta(n^{-3})$ and thus the cost of the projection within the model is $O({\rm poly}(n))$.  Thus if the measurement successfully projects onto the state of maximum Harmony then the computation can be implemented in polynomial time.

Next we need to show that the measurement at the end of the protocol can be prepared with high probability.  The method proposed in~\cite{GOSSET} is to use adiabatic state preparation, which is guaranteed to work in polynomial time because the minimum gap is inverse-polynomial.  In our context, we do not have the ability to perform an adiabatic sweep within our model so we instead use the Zeno-effect to emulate it.

First the state of maximum Harmony for $\mathcal{F}(0)$ is chosen by design to be a state of the form $\prod_{j=0}^{2n}a^\dagger_{(j,0),0} \ket{0}$~\cite{GOSSET}.  By assumption, this state can be prepared at no cost in the above model for a Harmonic quantum computer.  Now for any $s\in [0,1]$ we have that the state of maximum Harmony $\ket{\psi_{\max}(s)}$ can be chosen (by selecting an appropriate global phase as a function of $s$) to obey for an orthonormal set of instantaneous eigenvectors $\{\ket{\psi(s)}\}$ of $\mathcal{F}(s)$
\begin{align}
\left.\frac{\partial}{\partial \Delta} \ket{\psi_{\max}(s+\Delta)}\right|_{\Delta =0}  &= \sum_{\psi(s) \ne \psi_{\max}(s)} \frac{\bra{\psi(s)} \dot{\mathcal{F}}(s) \ket{\psi_{\max}(s)}}{\bra{\psi_{\max}(s)} \mathcal{F}(s) \ket{\psi_{\max}(s)} - \bra{\psi(s)} \mathcal{F}(s) \ket{\psi(s)}} \ket{\psi(s)},\nonumber\\
&:=\sum_{\psi(s) \ne \psi_{\max}(s)} \frac{\bra{\psi(s)} \dot{\mathcal{F}}(s) \ket{\psi_{\max}(s)}}{\gamma_{\psi,\psi_{\max}}}\ket{\psi(s)}.
\end{align}
which exists because the spectral gap is in $\Omega(n^{-3})$ from~\eqref{eq:gapbd}.  It then follows from Taylor's theorem that if we take $P^\perp(s)=(\openone - \ketbra{\psi_{\max}(s)}{\psi_{\max}(s)})$ then it is immediately clear from the fact that $P^{\perp}(s)\ket{\psi_{\max}(s)}=0$ that
\begin{equation}
\frac{\partial}{\partial \Delta} \bra{\psi_{\max}(s+\Delta)}P^{\perp}(s)\ket{\psi_{\max}(s+\Delta)}\Biggr|_{\Delta =0} = \bra{\dot{\psi}_{\max}(s)}P^{\perp}(s)\ket{\psi_{\max}(s)}+\bra{\psi_{\max}(s)}P^{\perp}(s)\ket{\dot{\psi}_{\max}(s)}=0.\label{eq:1deriv0}
\end{equation}
Similarly, it is easy to see that
\begin{equation}
\left.\frac{\partial^2}{\partial \Delta^2}\bra{\psi_{\max}(s+\Delta)}P^{\perp}(s)\ket{\psi_{\max}(s+\Delta)}\right|_{\Delta =0} = 2\bra{\dot{\psi}_{\max}(s)}P^{\perp}(s)\ket{\dot{\psi}_{\max}(s)}
\end{equation}

Thus we have that if the spectral gap is at least $\gamma(s)$ for all $s$ then
\begin{align}
\left|\left.\frac{\partial^2}{\partial \Delta^2 }\bra{\psi_{\max}(s+\Delta)}P^{\perp}(s)\ket{\psi_{\max}(s+\Delta)}\right|_{\Delta =0}\right|&=2 \left|\sum_{\psi(s),\psi'(s) \ne \psi_{\max}(s)} \frac{\bra{\psi(s)} \dot{\mathcal{F}}(s) \ket{\psi_{\max}(s)}}{\gamma_{\psi,\psi_{\max}}} \frac{\bra{\psi_{\max}(s)} \dot{\mathcal{F}}(s) \ket{\psi'(s)}}{\gamma_{\psi',\psi_{\max}}}\delta_{\psi,\psi'}\right|,\nonumber\\
& = 2 \left| \sum_{\psi(s) \ne \psi_{\max}(s)}\frac{\bra{\psi_{\max}(s)} \dot{\mathcal{F}}(s) P^{\perp}(s)\ketbra{\psi(s)}{\psi(s)} P^\perp(s) \dot{\mathcal{F}}(s) \ket{\psi_{\max}(s)}}{\gamma^2_{\psi,\psi_{\max}}}\right|\nonumber\\
& \le 2\frac{}{\gamma(s)^2}\left|\sum_{\psi(s)}\bra{\psi_{\max}(s)} \dot{\mathcal{F}}(s) P^{\perp}(s)\ketbra{\psi(s)}{\psi(s)} P^\perp(s) \dot{\mathcal{F}}(s) \ket{\psi_{\max}(s)}\right|\nonumber\\
& = \frac{2}{\gamma(s)^2} \left|\bra{\psi_{\max}(s)}\dot{\mathcal{F}}(s) P^{\perp}(s) \dot{\mathcal{F}}(s) \ket{\psi_{\max}(s)}\right|\nonumber\\
&\le \frac{2\|\dot{\mathcal{F}}(s)\|^2}{\gamma(s)^2}.\label{eq:2derivbd}
\end{align}
For any $\Delta>0$ we therefore have from~\eqref{eq:1deriv0},~\eqref{eq:2derivbd} and the fundamental theorem of calculus that
\begin{equation}
\bra{\psi_{\max}(s+\Delta)}P^{\perp}(s)\ket{\psi_{\max}(s+\Delta)} = \int_0^\Delta \int_{0}^{t} \frac{\partial^2}{\partial \delta^2}\bra{\psi_{\max}(s+\delta)}P^{\perp}(s)\ket{\psi_{\max}(s+\delta)}\mathrm{d}\delta \mathrm{d}t \le \max_s \frac{\Delta^2\|\dot{\mathcal{F}}(s)\|^2}{\gamma(s)^2}
\end{equation}

Now let us assume that we attempt to prepare the state $\ket{\psi_{\rm max}(1)}$ by uniformly sweeping over $s$ and taking $\Delta=1/r$ for $r$ steps.  This gives from the union bound
\begin{equation}
P_{\rm fail} \le r\max_s \left(\frac{\|\dot{F}\|}{r\gamma(s)} \right)^2 =\frac{1}{r}\max_s \left(\frac{\|\dot{F}\|}{\gamma(s)} \right)^2.
\end{equation}
Equation~\eqref{eq:univHarmony} gives that $\|\dot{\mathcal{F}}(s)\|$ is in $O({\rm poly}(n))$ and $\gamma(s)$ also is known to be in $\Omega({\rm poly}(n))$.  Therefore we have that
$r\in \Theta({\rm poly}(n))$ measurements in the eigenbasis of $\mathcal{F}(s)$ suffices to prepare the state.  Each such measurement must now, however, have probability of failure at most $1/r$ which necessitates a logarithmic number of repetitions given that by assumption the probability of success for measurement in the Harmonic quantum computer model is at least $2/3$.  Thus the total cost of preparing the state is polynomial and in turn all quantum circuits can be simulated within the model of computing.

Next we need to show the converse, specifically that circuit based quantum computers can simulate a Harmonic quantum computer within bounded error using a polynomial number of gate operations.  First, if the number of terms present in the Harmony operator is polynomial and each binding operator is itself representable as a row-computable $O({\rm poly}(n))$-sparse matrix then it follows that for all $s$, $\mathcal{F}(s)$ can be represented as a row-computable $O({\rm poly}(n))$-sparse matrix as well.

If $\mathcal{F}(s)$ is a row-computable $O({\rm poly})(n)$-sparse matrix then it follows~\cite{berry2007efficient} that, for any $s$, $e^{-i\mathcal{F}(s)}$ can be simulated within error $\epsilon$ using $O({\rm poly}(n/\epsilon))$ gates.  Thus by using phase estimation, we can simulate a measurement in the eigenbasis within error $\epsilon$ and probability of success greater than $2/3$ using $O(1/\epsilon)$ applications of this simulation.  Thus the measurement can be implemented within cost $O({\rm poly}(n/\epsilon))$ as required.

Next we have to be able to apply the binding operators to prepare the initial state.  This is potentially challenging as the binding operators need not be unitary.  This can be solved by noting that~\cite{jordan2014quantum}
\begin{equation}
e^{-i \pi (a^\dagger + a)/2} \ket{0} =-ia^{\dagger} \ket{0}.
\end{equation}
Thus we can prepare the state if we can apply the creation operator within error $\epsilon$ at cost $O({\rm poly}(n/\epsilon))$.  Since $a^\dagger$ can be represented as an $O({\rm poly}(n))$-sparse row computable matrix this is possible~\cite{berry2007efficient}.  Further, since the number of times this process must be repeated is $O({\rm poly}(n))$ it follows that we can also by decreasing the error tolerance in each individual state prep prepare the initial state within constant error using $O({\rm poly}(n))$ gate operations.
From this it follows that each step in a protocol involving a Harmonic quantum computer that satisfies our assumptions can be efficiently simulated on a quantum computer.  This proves the converse direction for the proof.
\end{proof}

A natural consequence of this theorem is that there also exist models of quantum Boltzmann machines on lattices that are polynomially equivalent to quantum computing.  Previous work has shown that in principle such Boltzmann machines are $\BQP$--hard to train and evaluate~\cite{kieferova2017tomography} but did not show that they are $\BQP$ complete~\cite{kieferova2017tomography}.  This work provides such a proof.

\subsection{Quantum Advantages for Classical Harmony Optimization}
If a Harmony operator is classical then in general it is reasonable to believe that a machine that implements Harmony optimization has no more power than a non-deterministic Turing machine.  While this is certainly true, polynomial improvements to the annealing methods discussed above can be achieved using quantum computers or potentially classes of quantum annealers.  We focus our attention on the case of quantum computers as the case for potential advantage is more clear in that setting.  Here we discuss the previous work of~\cite{somma2008quantum} which shows that polynomial advantages can be attained for classical annealing processes and argue that these speedups indeed will be polynomial for optimization problems such as those that we examined above.

The cost of simulated annealing depends strongly on the spectral gap of the associated Markov process.  If we consider a sequence of inverse temperatures, $\beta_k$ then let $\delta$ be the minimum spectral gap of the transition matrices corresponding to the temperatures $\beta_k$.  Also let $\gamma$ be the minimum gap between the maximum Harmony state and the state with next highest Harmony and let $\epsilon$ be a tolerable failure probability and let the process take place in a space with $D$ configurations.  Provided that the final value of $\beta$ obeys $\beta_f \in O(\gamma^{-1} \log(D/\epsilon^2))$ then the complexity of simulated annealing scales as~\cite{stroock2013introduction,somma2008quantum}
\begin{equation}
N_{SA} \in O\left(\frac{\|\mathcal{H}\|\log(D/\epsilon^2)}{\gamma\delta} \right).
\end{equation}
In practice, since $\|\mathcal{H}\|$ typically scales polynomially with the number of roles and the gap is on the order of $1$ the dominant contribution to the cost is from the gap of the Markov chain, $\delta$.

Quantum algorithms can be used to improve upon this.  The most natural way to do so is to replace the random process of choosing new configurations with a quantum walk on a bipartite graph with each subgraph corresponding to a different configuration for the system.  The purpose of the quantum walk is to accelerate the mixing time.  However, along the way the annealing process wherein $\beta$ is decreased requires projecting the quantum walk into an eigenstate of the walk operator at each step.  This process can be implemented using phase estimation and adds to the cost of the protocol.  The final complexity of the algorithm when accomodating for these issues is~\cite{somma2008quantum}
\begin{equation}
N_{QSA} \in O\left(\frac{\|\mathcal{H}\|^2\log^3(D/\epsilon^2)}{\gamma^2 \sqrt{\delta}} \right),
\end{equation}
which has quadratically better scaling with the gap of the Markov chain and at the price of worse scaling with the remaining parameters.  Given the fact that for our applications these parameters are expected to be exponentially smaller than the gap of the Markov chain, the improvements that quantum offers can be significant.

\section{Learning the Harmony Operator}
\label{sec:Learning}
The problem of Harmony optimization is in general a challenging problem.  It involves finding configurations that achieve maximal Harmony over the set of all possible occupations within the Fock space.  This problem is unlikely to be solvable either classically or quantumly even for classical Harmony operators.  This challenge can be clearly seen because binary satisfiability problems such as $3$-SAT can be mapped to a Harmony optimization problem.  Thus if we could efficiently maximize Harmony in general on either type of computer it would imply that either
$\BPP = \NP$ or $\BQP=\NP$, both of which are false under broadly accepted complexity-theoretic conjectures.  These complexity-theoretic results imply that we cannot expect greedy local optimizers to yield states of globally maximum Harmony.

The shortcomings of local optimizers can be sidestepped by using global optimizers such as simulated annealing \check{\cite{kirkpatrick1983optimization}}, quantum annealing \cite{kadowaki1998quantum}, or iterated local search {\cite{glover2006handbook}}.  Simulated annealing is a physics-inspired algorithm that aims to mimic annealing processes in metallurgy.  The idea behind the algorithm in our context is, given a particular Fock state $\ket{\phi}$, randomly alter the occupations using one of a set of predefined moves.  This move yields a new state $\ket{\phi'}$ which is accepted if the Harmony is improved but only rejected with probability proportional to $e^{\beta( H(\ket{\phi'}) - H(\ket{\phi}))}$ for some constant $\beta>0$ if the Harmony is not improved by the move.  This gives annealing the ability to escape from local optima while at the same time retaining many of the features of local optimization. This method is pursued in Section~\ref{sec:HMaxNumerics}.


\subsection{Quantum unsupervised learning of Harmony operators}

While classical Harmony functions may have a natural construction for the problem at hand, it is often difficult to find a unique quantum Harmony operator that is ideally suited for a given language processing task.  This naturally raises the possibility of inferring, from data from a language, a quantum Harmony operator that can then be used to determine whether a given sentence is grammatical.
In the next subsection we will propose addressing this problem by learning a quantum Harmony operator by supervised training. That is to say, the user is provided with an oracle that yields copies of quantum state vectors appended with a label that specifies whether the vector is grammatical or not.  We first pre-train a quantum Boltzmann machine \cite{wiebe2016quantum,amin2018quantum,kieferova2017tomography} to generate a surrogate for the data set,
using a learning algorithm that will assign weights to maximize Harmony for grammatical examples and minimize Harmony for ungrammatical examples.

This pre-training step can be done simply by applying the work of~\cite{kieferova2017tomography}.  The idea is to train a Boltzmann machine to generate a data set that is close to the distribution over the training data in terms of a natural statistical distance (or divergence).  The most natural figure of merit to use is the quantum relative entropy using either Golden-Thompson \cite{amin2018quantum,kieferova2017tomography} or relative entropy \cite{kieferova2017tomography} training.  Below we state the result for relative entropy training, but exactly the same result also holds for Golden-Thompson training which is better suited for cases where latent variables are used.

\begin{theorem}
Let $\ket{v_k}: k=1,\ldots, K$ be vectors in \check{$\mathbb{C}^{2^{n}}$} for positive integers $K$ and $n$ with $\rho = \frac{1}{K} \check{\sum_{k=1}^K} \kb{v_k}{v_k}$ and let ${\cH}(\omega)=\sum_{i=1}^D \omega_i \cH_i:\mathbb{R}^D \rightarrow \mathbb{C}^{2^n\times 2^n}$ be a map such that for any $\omega \in \mathbb{R}^D$, $\cH(\omega)$ is a quantum Harmony operator.  If we assume that we have an oracle $F$ that, when given $\omega\in \mathbb{R}^D$, yields $\sigma(\omega)=e^{H(\omega)}/{\rm Tr}(e^{H(\omega)})$, then the number of queries to $F$ and training examples needed to estimate the gradient of $S(\rho|\sigma(\omega))$ with respect to $\omega$ within error $\epsilon$ in the Euclidean-norm with probability greater than $2/3$ is in $\mathcal{O}(D^2/\epsilon^2)$.\label{thm:relent}
\end{theorem}
\begin{proof}
The proof is a direct consequence of Theorem 1 of~\cite{kieferova2017tomography}.
\end{proof}

This shows that you can learn a Harmony operator by training a quantum Boltzmann machine in an unsupervised manner \check{on a sample of sentences from the target language}.  Given a Harmony operator $\cH$ it is then easy to generate a possible parsing of the state.  You simply prepare an input state vector $\ket{v_{\rm test}}$ and apply quantum phase estimation to it using operator $\cH$.  The aim is then to find, subject to a Harmony threshold $\kappa$, an eigenvector $\ket{\sigma}$ of $\cH$ with eigenvalue $\sigma$ such that $|\bracket{\sigma}{v_{\rm test}}|$ is maximized subject to $\sigma \ge \kappa$.  The success probability of this procedure depends on the value of $\kappa$ and the overlap of the input state with the subspace of maximally Harmonic states.

\subsection{Quantum supervised learning of Harmony operators}

Quantum Harmony operators can also be learned in a supervised setting.  This form of quantum Boltzmann training has not been considered in the literature and can be applied to general quantum Boltzmann training processes.  For this reason, we will also include the possibility of hidden units in this form of training.  The weights on hidden units are in general harder to train in relative entropy training owing to the gradients no longer having a closed form.

Before starting, let us begin by introducing some notation.  First, we assume that the Hilbert space that the Harmony operator $\cH$ acts on is of the form $\mathbb{C}^{2^n}\otimes \mathbb{C}^{2^h}\otimes \mathbb{C}^2$ corresponding to subsystems for the input $\ket{v_k}$, the hidden units used to compute Harmony, and a label qubit.  In discriminative training of the Harmony operator, we need to constrain the first register to be $\ket{v_k}$.  We achieve this by adding a penalty to the Harmony operator conditioned on the input $\ket{v_k}$.  We call this Harmony operator $H_k$ and we denote the strength of these constraints $\lambda$.

\begin{equation}
H_k := \lambda \ket{v_k}\!\bra{v_k}\otimes \openone + \cH.
\end{equation}
The constraint is rigidly enforced by taking the limit as $\lambda\rightarrow \infty$.
Further, we can define a conditional Harmony operator $H_k'\in \mathbb{C}^{2^{h+1}\times 2^{h+1}}$ such that
\begin{equation}
[H'_k]_{x,y}:=\bra{v_k}\bra{\sigma_{k,x}} \cH \ket{v_k} \ket{\sigma_{k,y}}
\end{equation}
for a set of basis vectors $\ket{\sigma_{k,j}}$ \check{spanning the hidden-state space}.  We will choose these vectors ultimately to diagonalize $H_k'$.  We also define for any operator $f(k)$
\begin{equation}
\langle (f(k)) \rangle_k = \frac{{\rm Tr}~\left[f(k) \ket{v_k}\!\bra{v_k} \otimes e^{H_k'}\right]}{{\rm Tr}~e^{H'_k}},\qquad \mathbb{E}_k(\cdot) = \frac{1}{K} \sum_{k=1}^K f(k).\label{eq:expect}
\end{equation}
Finally, for notational simplicity we introduce a projector onto the label space that serves to test whether the label assigned to a vector by the Harmony operator is correct, which means that the label assigned to $\ket{v_k}$ is the value $\ell_k$ which is stored in the final qubit.
\begin{equation}
P_{\ell_k}  = \openone \otimes \openone \otimes \ket{\ell_k}\!\bra{\ell_k}.
\end{equation}
For example, if the sentence encoded in $\ket{v_k}$ is grammatical (maximally Harmonic) then $\ell_k=1$ and otherwise it is $0$.
A natural training objective for such discriminative training is the classification accuracy.  (Cross entropy could be considered, but the matrix logarithm needed in this makes it difficult to find an analytic form for the gradient of the training objective function.)  We define this function as
$$
\lim_{\lambda\rightarrow \infty} \frac{1}{K} \sum_{k=1}^K {\rm Tr}\left( \frac{P_{\ell_k} e^{H_k}}{ {\rm Tr}~ e^{H_k}}\right)
$$
The gradients of this objective function are given below.

\begin{theorem}
Let $\ket{v_k}: k=1,\ldots, K$ and $\ket{\ell_k}: k=1,\ldots, K$ be vectors in $\mathbb{C}^{2^{n}}$ and $\check{\mathbb{C}^2}$ for positive integers $K$ and $n$, let ${\cH(\omega)}=\sum_{i=1}^D \omega_i \cH_i:\mathbb{R}^D \rightarrow \mathbb{C}^{2^{n+h+1}\times 2^{n+h+1}}$ be a map such that for any $\omega \in \mathbb{R}^D$, $\cH(\omega)$ is a quantum Harmony operator. We then have that if $[P_{\ell_k}, \cH_i]=0$ for all $k$ and $i$ then
$$
\lim_{\lambda\rightarrow \infty}\partial_{\omega_i} \frac{1}{K} \sum_{k=1}^K {\rm Tr}\left( \frac{P_{\ell_k} e^{H_k}}{ {\rm Tr}~ e^{H_k}}\right) =  \mathbb{E}_k \left[ \langle P_{\ell_k}  \cH_i\rangle_k  - \langle P_{\ell_k} \rangle_k \langle \cH_i \rangle_k \right].
$$
\label{thm:deriv}
\end{theorem}
Note that Harmonic Grammar operators such as (\ref{eq:treeExample}) satisfy the conditions of this theorem also Boltzmann machines satisfy this theorem under the transformation $\mathcal{H} \mapsto -H$ where $H$ is the Hamiltonian operator used in the quantum Boltzmann machine.

\begin{proof}
From the product rule we have that
\begin{align}
\partial_{\omega_i} \frac{1}{K} \sum_{k=1}^K {\rm Tr}\left( \frac{P_{\ell_k} e^{H_k}}{ {\rm Tr}~ e^{H_k}}\right) &=
\frac{1}{K} \sum_{k=1}^K \left[ {\rm Tr}\left(\frac{P_{\ell_k} \partial_{\omega_i} e^{H_k}}{ {\rm Tr}~ e^{H_k}} \right) - {\rm Tr}\left(\frac{P_{\ell_k}e^{H_k}{\rm Tr}~\left(\partial_{\omega_i} e^{H_k} \right)}{\left({\rm Tr}~ e^{H_k}\right)^2} \right) \right].
\end{align}
First note that from Duhamel's formula, the cyclic property of the trace and the assumption that $[P_{\ell_k}, \cH_i]=0$ for all $k,i$
\begin{align}
{\rm Tr} (P_{\ell_k} \partial_{\omega_i} e^{H_k}) &= {\rm Tr} \left(P_{\ell_k}\int_{0}^1 e^{H_k s} (\partial_{\omega_i} H_k)e^{H_k (1-s)} \mathrm{d}s \right)\nonumber\\
&= {\rm Tr} \left(P_{\ell_k}  (\partial_{\omega_i} H_k)e^{H_k } \right)={\rm Tr} \left(P_{\ell_k}  (\partial_{\omega_i} \cH)e^{H_k } \right)
\end{align}
Using exactly the same argument
\begin{align}
{\rm Tr} ( \partial_{\omega_i} e^{H_k}) &= {\rm Tr} \left(  (\partial_{\omega_i} \cH)e^{H_k } \right)
\end{align}
Thus for any $\lambda\in \mathbb{R}_+$
\begin{equation}
\partial_{\omega_i} \frac{1}{K} \sum_{k=1}^K {\rm Tr}\left( \frac{P_{\ell_k} e^{H_k}}{ {\rm Tr}~ e^{H_k}}\right) = \mathbb{E}_k \left[ \frac{{\rm Tr} \left(P_{\ell_k}  (\partial_{\omega_i} \cH)e^{H_k }\right)}{ {\rm Tr}~ e^{H_k}}  - {\rm Tr}\left(\frac{P_{\ell_k}e^{H_k}{\rm Tr}~\left( (\partial_{\omega_i} \cH)e^{H_k } \right)}{\left({\rm Tr}~ e^{H_k}\right)^2} \right) \right].\label{eq:mainThm}
\end{equation}

While the above expression holds for any valid constraint penalty $\lambda$, we want to understand the performance of the Boltzmann machine in the limit where the strength of the penalty goes to infinity.  Fortunately, we can argue about the form of the eigenvalues and eigenvectors of each $H_k$ in this limit.  This can be achieved using degenerate perturbation theory.

We now recount the argument from degenerate perturbation theory for completeness.
Consider the state $\ket{v_k}\otimes \ket{\sigma_{k,j}}$ where $\ket{\sigma_{k,j}}$ is chosen to be an eigenstate of $H'_k = \sum_{x,y} \ket{\sigma_{k,x}}\!\bra{\sigma_{k,y}}\bra{v_k}\bra{\sigma_{k,x}} \cH \ket{v_k} \ket{\sigma_{k,y}}$.  We then have that
\begin{align}
(\lambda \ket{v_k}\!\bra{v_k} \otimes \openone) \ket{v_k}\ket{\sigma_{k,j}}=\lambda\ket{v_k}\ket{\sigma_{k,j}},
\end{align}
which implies that this is an eigenstate of the constraint operator if one neglects $\cH$ in $H_k$.  Formally, let us consider the true eigenvalue of the operator and assert a Taylor series in powers of $\lambda^{-1}$ for the eigenvalue and eigenvector of $H_k$.  Specifically, for a fixed eigenstate $\ket{\psi_{k,j}}$ with eigenvalue $E_{k,j}$ we clearly have by taking the limit as $\lambda^{-1}\rightarrow 0$ that under the assumption that the eigenvalues and eigenvectors are chosen to be differentiable functions of $\lambda$, $\ket{\psi_{k,j}} = \ket{v_k}\ket{\sigma_{k,j}}+O(1/\lambda)$ and $E_{k,j} = \lambda + E_{k,j}^1 +O(1/\lambda)$ so
\begin{align}
\lambda \ket{v_k}\!\bra{v_k}\otimes \openone \ket{v_k}\ket{\sigma_{k,j}} + \cH \ket{v_k}\ket{\sigma_{k,j}}&= \lambda \ket{v_k} \ket{\sigma_{k,j}}+E^1_{k,j} \ket{v_k} \ket{\sigma_{k,j}} + O(1/\lambda).
\end{align}
From taking the $O(1)$ component of this equation we see that $E_{k,j}^1 = \bra{v_k} \bra{\sigma_{k,j}} \cH \ket{v_k} \ket{\sigma_{k,j}}:=\sigma_{k,j}$.
Note that the choice of $\ket{\sigma_{k,j}}$ to be eigenvectors of $H'_k$ allows us to guarantee that the eigenvectors of $H_k$ can be expressed as a differentiable function of $\lambda^{-1}$.  Thus to leading order in $\lambda^{-1}$ we can see that

\begin{equation}
(\ket{v_k}\!\bra{v_k} \otimes \openone)H_k = \sum_j \ket{v_k}\ket{\sigma_{k,j}}\!\bra{v_k}\bra{\sigma_{k,j}} (\lambda+ \sigma_{k,j})
\end{equation}

From this we can reason about $e^{H_k}/{\rm Tr}~ e^{H_k}$ in this limit.
\begin{equation}
\lim_{\lambda\rightarrow \infty} \frac{e^{H_k}}{{\rm Tr}~e^{H_k}} = \lim_{\lambda \rightarrow \infty} \frac{\sum_j \ket{v_k}\ket{\sigma_{k,j}}\!\bra{v_k}\bra{\sigma_{k,j}} e^{\lambda +\sigma_{k,j}}+O(1/\lambda)}{\sum_j e^{\lambda +\sigma_{k,j}}+O(1/\lambda)}=\frac{\ket{v_k}\!\bra{v_k} \otimes e^{H'_k}}{{\rm Tr} (e^{H'_k})}.
\end{equation}
The result then follows from~\eqref{eq:mainThm} and~\eqref{eq:expect} after noting $\partial_{\omega_i} \cH = \cH_i$:
\begin{align}
&\mathbb{E}_k \left[ \frac{{\rm Tr} \left(P_{\ell_k}  (\partial_{\omega_i} \cH)e^{H_k }\right)}{ {\rm Tr}~ e^{H_k}}  - {\rm Tr}\left(\frac{P_{\ell_k}e^{H_k}{\rm Tr}~\left( (\partial_{\omega_i} \cH)e^{H_k } \right)}{\left({\rm Tr}~ e^{H_k}\right)^2} \right) \right]\nonumber\\
&\qquad= \mathbb{E}_k \left[ \frac{{\rm Tr} \left(P_{\ell_k}  \cH_i \ketbra{v_k}{v_k}\otimes e^{H_k' }\right)}{ {\rm Tr}~ e^{H_k'}}  - {\rm Tr}\left(\frac{P_{\ell_k} \ketbra{v_k}{v_k}\otimes e^{H_k'}{\rm Tr}~\left(  \cH_i \ketbra{v_k}{v_k} \otimes e^{H_k' } \right)}{\left({\rm Tr}~ e^{H_k'}\right)^2} \right) \right]\nonumber\\
&\qquad= \mathbb{E}_k \left[ \langle P_{\ell_k}  \cH_i\rangle_k  - \langle P_{\ell_k} \rangle_k \langle \cH_i \rangle_k \right]
\end{align}
\end{proof}
The above result shows that an elementary expression for the gradient exists that can be expressed in terms of constrained expectation values of the terms in the Harmony operator.

\subsection{Quantum complexity of learning Harmony operators}

We now turn to the computational complexity of a supervised learning procedure exploiting Theorem \ref{thm:deriv}.  \check{We note that Harmonic Grammar operators like (\ref{eq:treeExample}) obey the conditions of this theorem, although under a different decomposition of $\cH$ into a sum of operators $\cH_{k}$ than that explicitly given in (\ref{eq:treeExample}).}  In particular, the operators $n_{\cdot}$ given in that equation are not unitary.  This is because the number operators count the occupation of a particular mode in the Fock representation and as a result obey $n_{\cdot} \ket{0} = 0$, which is manifestly non-unitary because $0$ is not a unit vector but $\ket{0}$ is.  We can, however, address this by expressing $n_{(\cdot)}$ in the Pauli-basis.  In particular $n_{(\cdot)} = (\openone - Z_{(\cdot)})/2$ where $Z_{(\cdot)}$ is the Pauli-Z operation acting on the same mode.  Since $Z=Z^\dagger$ and $Z^2 =\openone$ after making this substitution it is clear that we can re-write the Harmony operator to conform to the the assumptions of~Theorem~\ref{thm:deriv}.

\begin{corollary}\label{cor:comp}
Under the assumptions of Theorem~\ref{thm:deriv}, with the further assumption that $\cH_i$ is unitary and Hermitian for each $\cH_i$ in the Harmony operator, and given access to a unitary oracle $F(k):\ket{0}\mapsto \ket{v_k}\ket{\ell_k}$ and a state preparation oracle $G_k$ that prepares copies of $e^{H_k'}/{\rm Tr}~e^{H_k'}$, the number of queries to these oracles that are needed to compute a vector with components $ \mathbb{E}_k \left[ \langle P_{\ell_k}  \cH_i\rangle_k  - \langle P_{\ell_k} \rangle_k \langle \cH_i \rangle_k \right]$ within error $\epsilon$  in the Euclidean distance with probability at least $2/3$ is in $O(D^2/\epsilon^2)$.  If the process $G_k$ is defined such that $G_k \ket{k} \ket{0} = \ket{k}\ket{\psi_k}$ where $\ket{\psi_k}$ is a purification of $e^{H_k'}/{\rm Tr}~e^{H_k'}$ then the query complexity can be reduced to $O(D^2/\epsilon)$
\end{corollary}
\begin{proof}
The algorithm for achieving this is constructive.  First, it is straightforward to see that with two queries to $F$ it is possible to construct a gate $U$ that marks the state $\ket{\ell_k}$---that is to say that $U\ket{\psi}=-\ket{\psi}$ if and only if $\psi = \ell_k$, and otherwise $U$ acts as the identity.  By applying the Hadamard test with this unitary and using a state generated by the oracle $G_k$ as input it is easy to see that you can sample from a random variable with expectation value $1/2+{\rm Tr}(\ketbra{v_k}{v_k} \otimes e^{H_k'} U)/2({\rm Tr} e^{H_k'})=1-{\rm Tr}(\ketbra{v_k}{v_k}\otimes e^{H_k'} P_{\ell_k})/2{\rm Tr} e^{H_k'}$.  The variance of this random variable is at most $1$ because the norm of any projector is at most $1$.  Similarly, because $\cH_i$ is unitary, we can use the Hadamard test in the same way to sample from a random variable with mean 
\begin{equation}
\frac{1}{2}\left(1+  {\rm Tr}\left( \frac{\ketbra{v_k}{v_k} \otimes e^{H_k'} \cH_i}{{\rm Tr}~e^{H_k'}} \right)\right)  = \frac{1}{2}  + \frac{\langle \cH_i \rangle_k}{2}
\end{equation} 
and variance at most $1$.  Finally, by applying $U\cH_i=\cH_i -2P_{\ell_k}\cH_i$ to the state yielded by $G_k$, we can sample from a random variable with mean $1/2 + \langle \cH_i\rangle_k/2 - \langle P_{\ell_k} \cH_i \rangle_k$ with variance at most $1$.  From the additive property of variance the number of samples needed to estimate each component of the gradient within error $\delta$ with probability at least $2/3$ is at most $O(1/\delta^2)$ repetitions of the circuit.  In order to guarantee that the error in the Euclidean norm is at most $\epsilon$ it suffices to take $\delta =\epsilon/D$.  The result then follows for the case where $G_k$ yields copies of a Gibbs state.

If a purified Gibbs state oracle is provided then by preparing a uniform superposition over the $K$ elements a state of the form $\frac{1}{\sqrt{K}} \sum_k \ket{k} \ket{\psi_k}$ can be prepared.  If we assume that $\ket{\psi_k} \in A\otimes B$ where $A$ is the Hilbert space corresponding to the domain of ${H}'_k$ and $B$ is an auxillary Hilbert space and if we define $K$ to be the Hilbert space used for the control register for $G_k$ then we see that
\begin{equation}
{\rm Tr}_{K,B} \left(\frac{1}{\sqrt{K}} \sum_k \ket{k} \ket{\psi_k}\right) = \mathbb{E}_k \left( \frac{e^{H'_k}}{{\rm Tr} e^{H'_k}} \right).
\end{equation}
Thus by repeating the above steps involving the Hadamard tests, we can create a unitary circuit such that the measurement of an individual qubit yield a random variables with means
\begin{equation}
\mathbb{E}_k \left( \frac{\langle \cH_i \rangle_k}{2}\right), \qquad \mathbb{E}_k\left(\frac{\langle \cH_i\rangle_k}{2} - \langle P_{\ell_k} \cH_i \rangle_k\right),
\end{equation}
and variances at most $1$.  Finally, by using amplitude estimation these means can be extracted within error $\epsilon/D$ using $O(D/\epsilon)$ applications of the above protocol~\cite{brassard2002quantum}.  This yields one component of the gradient and as there are $D$ components the total query complexity is $O(D^2/\epsilon)$ as claimed.
\end{proof}
The following lemma is well known and a proof can be found in~\cite{boyd2004convex}.
\begin{lemma}\label{lem:gradDescent}
Suppose that $f$ is a strongly convex function that obeys $(f(x')- f(x)-\nabla f(x)\cdot(x'-x))\|x'-x\|^{-1} \in [\mu/2,L/2]$ and achieves its global minimum at $x=x^*$. If the rate of descent $r$ is chosen such that $r=1/L$ then at iteration $t$ of gradient descent the distance from the optimum parameters obeys
\begin{align}
\|x(t)-x^*\| \leq \left(1-\frac{\mu}{L}\right)^t\|x(0)-x^*\|.
\end{align}
This implies that for every $\epsilon>0$ there exists $t\in O\left((L/\mu)\log\left(\|x(0)-x^*\|/\epsilon\right)\right)$ such that $\|x(t)-x^*\|\le \epsilon$.
\end{lemma}

Using Lemma~\ref{lem:gradDescent} we can then bound the number of iterations of gradient descent that we will need to find a local optima for the training objective function.

\begin{theorem}
\label{thm:Complexity}
Let $f(\omega):=\lim_{\lambda\rightarrow \infty}\frac{1}{K} \sum_{k=1}^K {\rm Tr}\left( \frac{P_{\ell_k} e^{H_k}}{ {\rm Tr}~ e^{H_k}}\right)$ such that $f$ satisfies the requirements of~Lemma~\ref{lem:gradDescent}.  If $\omega^*$ represents the optimal weights for the Harmony operator $\cH$ within a compact region where $f$ is strongly convex and we assume there exists $\mathcal{L}\in \mathbb{R}$ such that
$\|\nabla f(\omega) - \nabla f(\omega')\| \le \mathcal{L} \|\omega - \omega'\|~\forall~\omega,\omega'$ then under the assumptions of~Corollary~\ref{cor:comp} the number of queries to a method that prepares the training data and prepares the data to find $\omega'$ such that $\|\omega' - \omega^*\|\le \epsilon$ is in
 $$
 \widetilde{O}\left(\frac{D^2}{\epsilon^2} \left(\frac{\|\omega^*-\omega_0\|}{\epsilon} \right)^{\frac{\log(1+\mathcal{L}/L)}{ \log(L/(L-\mu))}} \right);
 $$
here we assume $\mathcal{L},L$ and $\mu$ are constants and $\widetilde{O}(\cdot)$ means $O(\cdot)$ but neglecting sub-polynomial factors.
\end{theorem}
\begin{proof}
Let $\widetilde{\nabla f}$ be the numerical approximation to the gradient taken at a point and let $\tilde{\omega}_p$ be the approximation to $\omega_p$ that arises due to inexact gradient calculation.  By assumption the update rule used in the gradient asccent algorithm is $\omega_{p+1} = \omega_p + \nabla f(\omega_p)/L$ and $\tilde{\omega}_{p+1} = \tilde{\omega}_p + \widetilde{\nabla f}(\tilde{\omega}_p)/L$ We then have that if $\omega_0 = \tilde{\omega}_0$ then for any $p>0$ we have from the triangle inequality that if the gradients are computed such that $\| \nabla f(\tilde{\omega}_{p}) - \widetilde{\nabla f}(\tilde{\omega}_p)\|\le \delta$
\begin{align}
\|\omega_{p+1} - \tilde{\omega}_{p+1}\| &\le \|\omega_{p} - \tilde{\omega}_{p}\|+\frac{1}{L} \|\nabla f(\omega_p) - \nabla f(\tilde{\omega}_p)\| + \frac{1}{L} \|\nabla f(\tilde{\omega}_p) - \widetilde{\nabla f}(\tilde{\omega}_p)\|\nonumber\\
&\le \left( 1 + \frac{\mathcal{L}}{L}\right)\|\omega_{p} - \tilde{\omega}_{p}\|+\frac{\delta}{L}.
\end{align}
It then follows inductively from the initial condition $\omega_0 = \tilde{\omega}_0$ that
\begin{equation}
\|\omega_{k} - \tilde{\omega}_{k}\| \le \frac{\delta}{L} \sum_{p=0}^{k-1} \left( 1 + \frac{\mathcal{L}}{L}\right)^p= \frac{\delta}{\mathcal{L}} \left(\left(1+\frac{\mathcal{L}}{L}\right)^k-1 \right) \le \frac{\delta}{\mathcal{L}} \left( 1 + \frac{\mathcal{L}}{L}\right)^k.
\end{equation}
We therefore also have from the triangle inequality that
\begin{equation}
\|\omega^* - \tilde{\omega}_k\| \le \|\omega^* - \omega_k\| + \|\omega_{k} - \tilde{\omega}_{k}\| \le \left(1 - \frac{\mu}{L} \right)^k\|\omega_0 - \omega^*\|+ \frac{\delta}{\mathcal{L}} \left( 1 + \frac{\mathcal{L}}{L}\right)^k.
\end{equation}
In order to ensure that the overall error is at most $\epsilon$ we choose both contributions to be at most $\epsilon/2$.  Elementary algebra then shows that it suffices to choose
\begin{align}
k &= \left \lceil \frac{\log \left(\frac{2\|\omega_0 - \omega^*\|}{\epsilon} \right)}{\log \left( \frac{L}{L-\mu}\right)} \right\rceil \le 1 +\frac{\log \left(\frac{2\|\omega_0 - \omega^*\|}{\epsilon} \right)}{\log \left( \frac{L}{L-\mu}\right)} \in O\left(\log\left(\frac{\|\omega^* - \omega_0\|}{\epsilon} \right) \right),\nonumber\\
\delta &= \frac{\epsilon \mathcal{L}}{2} \left(1 + \frac{\mathcal{L}}{L} \right)^{-1}\left(\frac{\epsilon}{2\|\omega^* - \omega_0\|} \right)^{\frac{\log(1+\mathcal{L}/L)}{ \log(L/(L-\mu))}}\in O\left( \frac{\epsilon^{1+{\frac{\log(1+\mathcal{L}/L)}{ \log(L/(L-\mu))}}}}{\|\omega^* - \omega_0\|^{\frac{\log(1+\mathcal{L}/L)}{ \log(L/(L-\mu))}}} \right),\label{eq:kdScale}
\end{align}
where we have assumed $\mathcal{L}, L$ and $\mu$ are the constants defined above.

If we use the result of Corollary~\ref{cor:comp} then the total number of iterations needed is $O(D^2/\delta^2)$.  However, the probability of success for this process is at least $2/3$.  This means that, if the probability of each derivative failing can be reduced to $1/3k$, then the probability of success of the protocol is at least $2/3$.  From the Chernoff bound, this can be achieved using $O(\log(k))$ repetitions of each gradient calculation.  Thus the total number of samples scales as
\begin{equation}
N_{\rm samp} \in O(k\log(k) D^2 /\delta^2)\subseteq \widetilde{O}\left( k D^2/\delta^2\right).\label{eq:sampScale}
\end{equation}
Hence it follows from~\eqref{eq:kdScale} and~\eqref{eq:sampScale} that  the overall number of samples needed for the entire protocol scales as
\begin{equation}
N_{\rm samp} \in \widetilde{O}\left(\frac{D^2}{\epsilon^2} \left(\frac{\|\omega^*-\omega_0\|}{\epsilon} \right)^{2\frac{\log(1+\mathcal{L}/L)}{ \log(L/(L-\mu))}} \right).
\end{equation}
\end{proof}

This result shows that the number of samples needed in the training process scales inverse polynomially with the target uncertainty, and this holds even under the worst case assumption that deviations in the gradient calculations lead to exponentially diverging solutions.  Also, as this result does not explicitly depend on the form of the function $f$ (apart from guarantees about its \check{convexity and smoothness}) the result also holds generally for \check{supervised} training of quantum Boltzmann machines.

\section{Conclusion}
In this work we examine the question of how we can best fit certain problems in computational linguistics onto quantum computers.  In doing so, we propose a new formalism for representing language processing called Fock-space representations that have the advantage of being easily encoded in a small number of qubits (unlike tensor-product representations).  We then develop a formalism for harmonic grammars in this representation and show how to generalize it beyond the case of classical grammars.  This quantum case we find is related to quantum error correcting codes and furthermore cannot be efficiently solved on a classical computer unless $\P=\BQP$ meaning that our results can potentially offer exponential speedups for evaluating quantum Harmony operators unless classical computers are at most polynomially weaker than quantum computers.  In doing so, we also provide new methods for training quantum Boltzmann machines that may be of value independent of the current applications to language processing.  Finally, we illustrate the utility of Fock-space representations by showing how they, in concert with classical optimization methods, can be used to parse sentences in relatively complicated grammars very quickly on ordinary computers.

Looking forward there are a number of questions that are revealed by our work.  Perhaps the most evocative question is that of whether quantum harmonic grammars can provide better models for classes of natural language than classical grammars can.  Such work is challenging because these quantum models require exponential time to train on classical computers.  Thus testing the impact that quantum models will have for classical language processing will need to wait for quantum computing to mature.

Further work is also needed to understand the impact that Fock-space representations will have for computational linguistics using classical computers.  Based on our numerical studies on toy grammars, we have seen that this approach seems to be quite promising as a means to do language processing but it remains to be tested whether these methods will also perform well when applied to natural language.  It is our hope that while these ideas were borne out of a desire to investigate the role that quantum computing may play in computational linguistics, they will nonetheless feed back into the classical community and provide new quantum inspired methods for machine translation and natural language processing.

\bibliographystyle{unsrt}
\bibliography{biblio}

\end{document}